\documentclass[11pt,letterpaper]{article}

\usepackage[letterpaper,margin=1in]{geometry}
\usepackage[T1]{fontenc}
\usepackage[utf8]{inputenc}
\usepackage{times}
\usepackage{fullpage}
\usepackage{microtype}
\usepackage{setspace}
\usepackage{amsmath,amssymb,amsthm,mathtools}
\usepackage{xparse}
\usepackage{xcolor}

\usepackage[ruled,vlined,linesnumbered]{algorithm2e}
\SetKw{Break}{break}
\SetKwData{Null}{null}
\SetKwProg{WithProb}{\normalfont with probability}{}{end}

\usepackage{array}
\usepackage{booktabs}
\usepackage{diagbox}
\usepackage{multirow}
\usepackage{tabularx}
\usepackage{threeparttable}
\newcolumntype{Y}{>{\raggedright\arraybackslash}X}

\usepackage{enumitem}
\setlist{topsep=0.5em,itemsep=0.2em,parsep=0pt,partopsep=0pt}

\usepackage[
  colorlinks=true,
  linkcolor=blue!55!black,
  citecolor=blue!55!black,
  urlcolor=blue!55!black,
  pdfauthor={},
  pdftitle={Frequency Moments Beyond Equality: Streaming Cosine Density Moments}
]{hyperref}

\allowdisplaybreaks
\newtheorem{theorem}{Theorem}[section]
\newtheorem{lemma}[theorem]{Lemma}

\theoremstyle{definition}

\theoremstyle{remark}

\newtheorem{remark}[theorem]{Remark}

\newcommand{\abs}[1]{{\left | #1 \right |}}

\newcommand{\E}{\mathbf{E}}
\newcommand{\Var}{\mathbf{Var}}
\newcommand{\eps}{\epsilon}

\renewcommand{\Pr}{\mathbf{Pr}}

\newcommand{\R}{\mathbb{R}}

\newcommand{\loc}{\text{local}}
\NewDocumentCommand{\dl}{m o}{%
  \IfNoValueTF{#2}{d_{#1,\loc}}{d_{#1,\loc}^{#2}}%
}
\NewDocumentCommand{\Dl}{m o}{%
  \IfNoValueTF{#2}{D^{(#1)}_{\loc}}{D_{#1,\loc}^{#2}}%
}

\makeatletter
\let\oldnl\nl
\newcommand{\nonl}{\renewcommand{\nl}{\let\nl\oldnl}}
\makeatother

\newcommand{\ip}[2]{\left\langle #1,#2 \right\rangle}
\newcommand{\norm}[1]{\left\lVert #1 \right\rVert}

\newcommand{\diag}{\operatorname{diag}}
\newcommand{\rank}{\operatorname{rank}}
\newcommand{\tr}{\operatorname{tr}}
\newcommand{\Index}{\textnormal{\textsc{Index}}}
\newcommand{\1}{\mathbf{1}}
\newcommand{\ind}{\mathbf{1}}
\newcommand{\cP}{\mathcal{J}_{+}}
\newcommand{\cN}{\mathcal{J}_{-}}
\newcommand{\Gzero}{\mathcal{G}_0}
\newcommand{\Gone}{\mathcal{G}_1}
\newcommand{\Mp}{M_p}
\newcommand{\stirling}[2]{\genfrac{\{} {\}}{0pt}{}{#1}{#2}}

\newcommand{\qinsays}[2][]{}

\title{Frequency Moments Beyond Equality: Streaming \\ Cosine Density Moments}

\author{
	Qin Zhang \\
	Computer Science Department\\
	Indiana University\\
	\texttt{qzhangcs@iu.edu}
}

\date{}

\begin{document}
\hypersetup{pageanchor=false}

\begin{titlepage}
\thispagestyle{empty}

\maketitle

\begin{abstract}
For a stream of nonzero vectors $x_1,\ldots,x_n\in\mathbb{R}^d$, let $u_i=x_i/\|x_i\|_2$. We define the cosine density of the $i$-th stream element by $D_i:=\sum_{j\in[n]}\langle u_i,u_j\rangle$ and study the density moments $M_p:=\sum_{i\in[n]}D_i^p$ in both the signed- and nonnegative-cosine regimes. These quantities are similarity-aware analogues of classical frequency moments: replacing cosine similarity by equality (that is, $D_i = \sum_{j\in[n]} \1\{u_j = u_i\}$) gives $M_p=F_{p+1}$ and, in particular, $M_{-1}=F_0$, the number of distinct elements. We give one-pass streaming algorithms and lower bounds whose dependence on $d$ is tight up to logarithmic factors for $p=-1$, $p=1$, and even $p \ge 2$ (the latter under signed cosine), and whose exponents differ by less than $1/2$ for odd $p \ge 3$ under nonnegative cosine. Our results thus extend several fundamental statistics from the classical data stream literature to cosine similarity, a widely used measure for comparing vector embeddings in modern AI systems.

The main challenge in proving a space lower bound for nonnegative cosine is to eliminate unwanted contributions without relying on pairs of opposite vectors. We address this through a construction that we call \emph{equal-sum moment isolation}: two insertion-only prefixes have the same cardinality and vector sum, and a finite-difference comparison cancels the baseline shared by all their vectors while isolating the desired higher-order signal. This proof framework may be useful for other insertion-only streaming lower bounds, where direct cancellation is not possible.
\end{abstract}

\end{titlepage}

\setcounter{page}{1}
\setcounter{tocdepth}{1}
\clearpage

\pagenumbering{arabic}
\setcounter{page}{1}
\hypersetup{pageanchor=true}

\section{Introduction}
\label{sec:intro}

For more than four decades, the theory of data streams has largely treated the data universe as discrete: two items are either equal or different. This abstraction underlies fundamental statistics such as the number of distinct elements, the self-join size, and the frequency moments~\cite{FM85,AMS99}. Modern AI systems, however, increasingly represent text, images, audio, and other unstructured objects as high-dimensional vectors, where similarity is real-valued rather than binary. This shift changes the geometry of the data universe and raises the question of how classical streaming statistics should be defined and computed in this new setting.

As an example, consider two streams of $n$ pairwise distinct documents. Suppose that the documents in the first stream have unrelated content, whereas those in the second stream are paraphrases of the same underlying content. Because every document is distinct, the $q$-th frequency moment
$F_q=\sum_{a} f_a^q$, where $f_a$ denotes the frequency of universe element $a$, equals $n$ for both streams for every $q\geq 0$. Semantically, however, the
first stream is highly diverse, while the second is almost entirely
redundant. Exact matching therefore treats superficial differences as if they reflected meaningful diversity.  This problem is becoming increasingly serious with the rapid growth of AI-generated content, especially content produced through summarization, paraphrasing, translation, and the recombination of existing material.

A natural approach is to replace the exact frequency of an item by its total similarity mass. Let $[n] := \{1, \ldots, n\}$. Given a dataset
$V=(v_1,\ldots,v_n)$ and a similarity function ${sim}(\cdot, \cdot)$, define, following~\cite{LZ26b},  the \emph{similarity density} (or \emph{soft frequency})  of $v_i$ as $D_i:=\sum_{j\in[n]}sim(v_i,v_j)$ and the $p$-th density moment $M_p$ of $V$ as $M_p:=\sum_{i\in[n]}D_i^p$. When ${sim}$ is equality, we have $D_i = f_{v_i}$ and 
$$
M_p(V) = \sum_{i\in[n]}f_{v_i}^p = \sum_{a:f_a>0}f_a^{p+1} = F_{p+1}.
$$
Thus, $M_{-1}$, called the \emph{diversity index} in~\cite{LZ26b}, generalizes the distinct-elements count $F_0$, while $M_1=\sum_{i,j\in[n]}sim(v_i,v_j)$ is a soft analogue of
the self-join size $F_2$ and measures the dataset's total similarity.  Higher moments capture the \emph{skew} of the similarity structure: for a fixed total similarity $M_1$, the moment $M_p$ with $p>1$ is largest when the similarity mass concentrates on a few dense regions, such as a small number of heavily paraphrased items, and smallest when the similarity mass is spread evenly across the dataset.\footnote{Concretely, for nonnegative densities with a fixed sum $M_1$, Hölder’s inequality gives $M_p \ge n^{1-p}M_1^p$ for $p>1$, with equality when all densities are equal; $M_p$ grows as the densities become more unequal.} This mirrors the role of $F_p$ for $p>2$ in measuring the skew of a frequency vector and in detecting heavy hitters. The density moments therefore describe the \emph{shape} of the similarity distribution, from its diversity ($p=-1$) through its total mass ($p=1$) to its concentration ($p>1$).

Previous approaches to similarity-aware statistics largely attempt to recover a latent discrete universe. They assume that the items form well-separated groups~\cite{Zhang15,CZ16,CZ18,Zhang25}, or parameterize the discrepancy between the observed data and an unknown ground truth~\cite{LZ26}. Such assumptions may be restrictive in embedding spaces, where semantic regions can overlap and need not admit clear boundaries. Density moments avoid an explicit clustering or ground-truth assumption, and instead aggregate the full, continuously valued similarity structure.

Recent work~\cite{LZ26b} formulates similarity-aware statistics in terms of a weighted similarity graph, where data items correspond to vertices and pairwise similarities define edge weights. It shows that, for a general pairwise similarity function, {\em no} one-pass sublinear-space algorithm can approximate the corresponding density moments within a constant factor for $M_{-1}$ and $M_1$, or within a factor of $\Theta_p(n^{(p-1)/2})$ for $M_p$ when $p>1$.
The general similarity model, however, can encode essentially arbitrary weighted graphs. This raises a natural question: {\em can we design one-pass sublinear-space algorithms for natural similarity measures with additional structure?}

We focus on the cosine similarity $sim(x, y) = {\langle x/\norm{x}_2, y/\norm{y}_2 \rangle}$, a widely used similarity measure for comparing vector embeddings in modern AI systems~\cite{RG19,RKH+21}. For normalized vectors, cosine similarity equals the inner product. Thus, if $A\in\mathbb{R}^{n\times d}$ contains the normalized vectors as its rows, then the similarity matrix is the Gram matrix $G=AA^\top$, which is positive semidefinite and has rank at most $d$. Cosine similarity graphs therefore form a structured subclass of general weighted similarity graphs. Nevertheless, estimating a density moment requires aggregating similarity information over all pairs of vectors, and it is not immediately clear whether the algebraic structure of cosine similarity is sufficient to circumvent the lower bounds for general similarity functions. 

In this paper, we answer this question affirmatively. We design one-pass streaming algorithms for vector data under cosine similarity and establish space upper and lower bounds whose dependence on $d$ is tight up to logarithmic factors for $p=-1$, $p=1$, and even $p \ge 2$ (the latter under signed cosine), and whose exponents differ by less than $1/2$ for odd $p \ge 3$ under nonnegative cosine. As mentioned, under equality, these statistics correspond respectively to distinct elements $F_0$, the self-join size $F_2$, and the higher frequency moments $F_{p+1}$. Thus, our results provide similarity-aware, one-pass analogues of fundamental statistics in the classical data stream literature, for a similarity measure of central importance in modern AI systems.

We consider two cosine regimes. In the \emph{signed-cosine regime}, pairwise cosine similarities may be either positive or negative. In the \emph{nonnegative-cosine regime}, all pairwise inner products are nonnegative, which also ensures that every density is nonnegative. Under signed cosine, densities can be zero or negative. For even $p$, this causes no difficulty, since $M_p = \sum_i |D_i|^p \ge 0$. The odd exponent $p = 1$ is also unaffected: given a dataset $\{v_1, \ldots, v_n\}$, by bilinearity, $M_1 = \sum_{i,j} \mathrm{sim}(v_i, v_j) = \big\|\sum_i v_i/\|v_i\|_2\big\|^2 \ge 0$, so negative similarities simply lower the total similarity that $M_1$ measures. The remaining cases are different: $M_{-1}$ may be undefined, and for odd $p \ge 3$, $M_p$ itself can be negative, so a multiplicative approximation, which presumes a nonnegative target, is no longer meaningful. Under nonnegative cosine, every density is at least $1$, so $M_p$ is well defined and positive for every real $p$. The nonnegative-cosine regime is also common in practice: it holds automatically for nonnegative representations such as bag-of-words, TF–IDF, histogram features and for post-ReLU activations, and more generally whenever the vectors lie in a cone of half-angle at most $45^\circ$ around a common direction. Learned embeddings are known to concentrate in a narrow cone around a large shared mean direction~\cite{Ethayarajh19}, making pairwise cosine similarities generally positive.

We are particularly interested in the setting $n\gg d$, where the stream length far exceeds the representation dimension. This setting commonly arises in modern large-scale embedding applications, where representation dimensions are typically in the hundreds or low thousands, while the number of vectors may reach billions or even trillions. Our goal is therefore to obtain streaming algorithms whose space depends primarily on $d$.

At a high level, our results show that the algebraic structure of cosine similarity makes this goal possible. The positive-semidefinite, low-rank structure of the cosine Gram matrix can be exploited to obtain one-pass dimension-dependent streaming algorithms, in sharp contrast to the impossibility results for general similarity functions. At the same time, this structure creates new challenges for proving lower bounds, especially under nonnegative cosine, where cancellation by opposite vectors is unavailable. Our main lower-bound construction addresses this difficulty using equal-sum gadgets that localize the effect of each encoded bit, together with a problem-specific finite-difference argument that isolates the desired moment contribution.

\vspace{2mm}
\noindent{\bf Our Results.\ }
For simplicity, we fix the error probability at $\delta=0.01$. One word stores one input coordinate or intermediate scalar.  We write $O_c(\cdot)$ and $\Omega_c(\cdot)$ when the hidden constant may depend on $c$, and use $\widetilde{O}(\cdot)$ and $\widetilde{\Omega}(\cdot)$ to suppress logarithmic factors.  Formal definitions of the streaming model, approximation guarantees, and other conventions are given in Section~\ref{sec:preliminaries}.

For a stream of nonzero vectors $x_1,\ldots,x_n\in\mathbb{R}^d$, let $u_i:={x_i}/{\|x_i\|_2}$. Let $D_i:=\sum_{j \in [n]} \langle u_i,u_j\rangle$ and $M_p:=\sum_{i \in [n]} D_i^p$.
Our first result characterizes the complexity of the diversity index.

\begin{theorem}[Theorems~\ref{thm:diversity-main} and~\ref{thm:quadratic-lower-bound}]
	In the nonnegative-cosine regime, there is a one-pass streaming algorithm that outputs a $(1+\epsilon)$-approximation to $M_{-1}$ using $O\left(d^2\epsilon^{-2}+d\log n\right)$ words of space. 
	Conversely, for every constant $\epsilon \in (0, 1/5)$, every one-pass $(1+\epsilon)$-approximation algorithm requires $\Omega(d^2)$ bits of space.
\end{theorem}

For completeness, note that by convention, we define $M_0$ to be the stream length, while $M_1$ admits a simple exact $O(d)$-space algorithm, whose dependence on $d$ is optimal.

\smallskip

We now turn to higher density moments, beginning with an essentially complete characterization of the signed-cosine setting for even $p$.

\begin{theorem}[Theorems~\ref{thm:even-higher-moments},  \ref{thm:second-density-moment}, and~\ref{thm:signed-lower-bound}]
	Fix an even integer $p\geq 2$. There is a one-pass streaming algorithm that outputs a $(1+\epsilon)$-approximation to $M_p$ in the signed-cosine regime using $O_p\left(d^{p/2}\epsilon^{-2}\right)$ words of space.
	Conversely, every one-pass $(1+\epsilon)$-approximation algorithm requires
	$\widetilde{\Omega}_{p,\epsilon}(d^{p/2})$ bits of space for every constant $\epsilon \in (0,1)$ when $p \ge 4$, and $\Omega(d)$ bits for every constant
	$\epsilon \in (0,1/5)$ when $p = 2$.
\end{theorem}

Finally, we consider higher moments under the nonnegative-cosine constraint. Positivity permits algorithms for arbitrary real exponents, but makes strong lower bounds considerably more difficult.  Establishing this superlinear lower bound is the paper's main technical contribution. 

\begin{theorem}[Theorems~\ref{thm:approx-higher-moments} and~\ref{thm:main}]
	\label{thm:nonnegative-main}
	For every fixed real number $p>2$, there is a one-pass streaming algorithm that outputs a $(1+\epsilon)$-approximation to $M_p$ in the nonnegative-cosine regime using $\widetilde{O}_p\left(
	d^{p/2}\epsilon^{-4+4/p}\right)$ words of space.
	Moreover, if $p > 2$ is an integer, then there exists a constant $\epsilon_p>0$ such that every one-pass $(1+\epsilon_p)$-approximation algorithm requires $\Omega_p\left(d^{\rho_p}\right)$ bits of space, where $\rho_p :=\frac{p^2(p-1)}{2(p^2-2)}$.
\end{theorem}
Note that for every $p > 2$, $\rho_p \in (\frac{p-1}{2}, \frac{p}{2})$,
and hence the lower-bound exponent is less than $1/2$ away from the upper-bound exponent $p/2$.

\smallskip

Our main results are summarized in Table~\ref{tab:main-results}.

\begin{table*}[t]
	\centering
	
	\small
	\setlength{\tabcolsep}{4pt}
	\renewcommand{\arraystretch}{1.18}
	
	\begin{tabularx}{\textwidth}{
			@{}
			|p{0.14\textwidth} |
			p{0.15\textwidth}
			p{0.16\textwidth}
			p{0.20\textwidth}
			Y|
			@{}
		}
		
		\hline
		Moment
		& Cosine regime
		& Guarantee
		& Upper bound
		& Lower bound
		\\
		\hline
		
		$M_{-1}$
		& Nonnegative
		& $(1+\eps)$-approx
		&
		$\displaystyle
		O\!\left(
		d^2\epsilon^{-2}+d\log n
		\right)$
		&
		$\Omega(d^2)$
		\\
		
		$M_1$
		& Signed
		& Exact
		&
		$O(d)$
		&
		$\Omega(d)$
		\\
		
		$M_p$, even $p\geq 2$
		& Signed
		& $(1+\eps)$-approx
		&
		$\displaystyle
		O_p\!\left(d^{p/2}\epsilon^{-2}\right)$
		&
		$\displaystyle
		\widetilde{\Omega}_p
		\!\left(d^{p/2}\right)$, $\Omega(d)$ for $p=2$
		\\
		
		$M_p$, odd $p \ge 3$
		& Nonnegative
		& $(1+\eps)$-approx
		&
		$\displaystyle
		\widetilde{O}_p\!\left(
		d^{p/2}\epsilon^{-4+4/p}
		\right)$
		&
		$\Omega_p(d^{\rho_p})$, where
		$\rho_p = \frac{p^2(p-1)}{2(p^2-2)}$
		\\
		
		\hline
	\end{tabularx}
	\caption{
		Main results for cosine density moments.
		Upper bounds are measured in words and lower bounds in bits.
		All algorithms use one pass.
		The algorithm for $M_1$ is exact, whereas the algorithms for
		$M_{-1}$ and $M_p$ return $(1+\eps)$-approximations with
		constant error probability.
		Lower bounds for $M_{-1}$ and $M_p$ hold for sufficiently small constant $\eps$ and
		constant success probability; the $M_1$ lower bound is for exact
		computation.  The nonnegative-cosine upper bound extends to every fixed real $p>2$,
		while the corresponding lower bound holds for every fixed integer
		$p>2$.
	}
	\label{tab:main-results}
\end{table*}

\subsection{Related Work}
\label{sec:related}
We review a few research directions that are closely related to this work.

\vspace{2mm}
\noindent{\bf Classical Data Streams.}
The distinct elements problem $F_0$ and the frequency moments
$F_q\ (q > 0)$ are among the foundational problems in the data stream model~\cite{FM85,BJK+02,KNW10,AMS99}, to cite just a few. A long line of work has developed increasingly tighter upper and lower bounds for estimating frequency moments, especially in the high-moment regime~\cite{BGKS06,IW05,W04,LW13,GW18}.
Our density moments extend these classical statistics from equality to real-valued similarity. For cosine similarity, however, the density of a vector depends on its interaction with the entire stream rather than on a single coordinate of a discrete
frequency vector. Consequently, classical sketches for frequency moments do not directly apply.

\vspace{2mm}
\noindent{\bf  Streaming with Noisy Data.}
Prior work on streaming data with near-duplicates has studied robust distinct elements estimation and distinct sampling under geometric separation, well-shapedness, or related assumptions on the data~\cite{Zhang15,CZ16,CZ18,Zhang25}. Another framework assumes an unknown noiseless ground-truth dataset and parameterizes the guarantees by the discrepancy between the observed and ground-truth data~\cite{LZ26}. 

The closest predecessor to this paper studies degree moments, the diversity index, diversity sampling, and degree-moment sampling on general similarity graphs in the node-arrival model~\cite{LZ26b}. It gives constant-pass sublinear-space algorithms, but also proves strong $\Omega(n)$ one-pass lower bounds for general similarity functions, which can encode an arbitrary weighted graph. A follow-up work~\cite{Zhang26c} shows that even two passes require linear space for the diversity index under general similarity functions. This paper investigates how additional structure changes the one-pass complexity, giving bounds that depend primarily on the representation dimension $d$ rather than on the stream length $n$.  The lower bounds of~\cite{LZ26b,Zhang26c} are proved for general similarity functions and do not transfer to our setting, where hard instances must be realized by unit vectors in $\mathbb{R}^d$ under cosine similarity. 

A companion work~\cite{Zhang26b} studies the Gaussian kernel in fixed dimension and gives one-pass sublinear-space algorithms for the diversity index and the density moments. Its techniques exploit geometric locality and spatial decay and are substantially different from those developed here. The resulting space bounds are {\em exponential} in the vector dimension $d$ for the diversity index and polynomial in the stream length $n$ for the $p$-th density moment when $p>1$. One way to see the difference is that a Gaussian kernel matrix is also positive semidefinite, but it is generally full rank, so it lacks the low-rank structure that cosine similarity has.

\vspace{2mm}
\noindent{\bf Sublinear-Time Estimation of Degree Moments.}
Degree moments of graphs have also been studied in the sublinear-time query model~\cite{Feige04,GR08,GRS11,ERS17}. These works assume oracle access to a fixed graph through operations such as degree and neighbor queries, and try to minimize the number of such queries. Our setting is
different: vector nodes arrive sequentially, the edges of the similarity graph are implicit, and a similarity can be evaluated only when the corresponding nodes are available in memory. Thus, the query access available to sublinear-time algorithms is not present in our streaming model, and their techniques and bounds do not directly transfer.

\section{Technical Overview}
\label{sec:overview}

For a stream of normalized vectors $P=(u_1,\ldots,u_n)$, let $P_t=(u_1,\ldots,u_t)$ denote its prefix of length $t$, and let $s_t=\sum_{i\in[t]}u_i$. The density of a point $u_i$ in this prefix is then $D_i(t)=\ip{u_i}{s_t}.$ For the completed stream, we abbreviate $s=s_n$ and $D_i=D_i(n)$.

We begin with a collection of upper bounds and then turn to the lower bounds, which constitute the main technical contribution of the paper.

\subsection{A Map of Upper Bounds}
\label{sec:overview-ub}

\vspace{2mm}
\noindent{\bf Algebraic Identities and Signed Cosine.\ }
The simplest result is the exact $O(d)$-word algorithm for $M_1$: it maintains only $s$, since $M_1=\norm{s}_2^2$. More generally,
$
M_p=\left\langle \sum_{i\in[n]}u_i^{\otimes p},s^{\otimes p}\right\rangle,
$
so storing the degree-$p$ symmetric tensor gives an exact $O_p(d^p)$-word algorithm for every integer $p\geq1$. For an even exponent $p=2k$, we can reduce the tensor degree by half. Let $\psi_k:\mathbb R^d\to\mathbb R^{\binom{d+k-1}{k}}$ be the normalized degree-$k$ monomial feature map defined in Section 6.1. It satisfies $\langle\psi_k(v),\psi_k(w)\rangle=\langle v,w\rangle^k$, and hence,
$$
M_{2k}=\sum_{i\in[n]} \ip{\psi_k(u_i)}{\psi_k(s)}^2.
$$
An AMS sketch~\cite{AMS99} in this lifted space estimates the last sum using $O_p(d^{p/2}\eps^{-2})$ words (Theorem~\ref{thm:even-higher-moments}). This argument works even when individual densities are negative, since the exponent is even.

\vspace{2mm}\noindent{\bf Nonnegative Cosine.\ }
For a prefix stream of length $t$ with Gram matrix $G_t$ and density matrix $\Delta_t=\diag(D_1(t),\ldots,D_t(t))$, consider
$
B_t=\Delta_t^{-1/2}G_t\Delta_t^{-1/2}.
$
The matrix $B_t$ is positive semidefinite and similar to a stochastic matrix, and has rank at most $d$. Its eigenvalues therefore lie in $[0,1]$, one eigenvalue equals $1$, and $\tr(B_t)=M_{-1}(P_t)$. This gives the diversity--rank principle (Lemma~\ref{lem:rank-bound}):
$$
1\leq M_{-1}(P_t)\leq \rank(G_t)\leq d.
$$

Nonnegativity is also useful algorithmically in a second way: every prefix density $D_i(t)$ is nondecreasing in $t$, since $D_i(t)=D_i(t-1)+\ip{u_i}{u_t}\ge D_i(t-1)$ for every $t>i$. Our diversity estimator assigns vector $i$ a fixed random value $r_i \sim \mathrm{Unif}(0,1)$ and keeps it in the sample as long as $r_i \le \min\{1, \lambda/D_i(t)\}$, where $\lambda = \Theta(\epsilon^{-2})$. Because of this monotonicity, a rejected vector never becomes eligible again and can safely be discarded. The expected number of stored vectors is at most $\lambda M_{-1}(P_t) \le \lambda d$; a union bound over the $n$ time steps caps it at $O(\lambda d + \log n)$ with high
probability, and each stored vector takes $O(d)$ words, giving the one-pass $O(d^2\epsilon^{-2} + d\log n)$-word upper bound of Theorem~\ref{thm:diversity-main}.

For every fixed real $p>2$, nonnegative cosine gives a moment-ratio bound. Since $D_i\leq \sqrt{M_1}$ for every $i$, $M_{2p}\leq M_1^{p/2}M_p.$ Together with H\"older's inequality $n^{1-p} M_1^p \le M_p$, Cauchy--Schwarz $n^2\leq M_1M_{-1}$,  and the diversity--rank principle $M_{-1}\leq \rank(G)$, this gives
$$
\frac{nM_{2p}}{M_p^2}\leq \rank(G)^{p/2}\leq d^{p/2}.
$$
Consequently, uniformly sampling about $d^{p/2}/\eps^2$ indices is enough to estimate $M_p$. The sampled densities are {\em not} known when the corresponding points arrive, and explicitly storing all sampled vectors would use too much space. We instead take a linear $F_p$-sketch $\Pi$ of the sampled density vector $y$, whose coordinates are the densities $D_i$ of the sampled vectors, and lift each scalar sketch counter to a $d$-dimensional vector. If the resulting matrix is $Y$, then at the end of the stream $Ys=\Pi y.$ Thus, the desired sketch is recovered only after the final sum $s$ is known. Combining uniform sampling with this lifted sketch gives the $\widetilde O_p(d^{p/2}\eps^{-4+4/p})$-word algorithm of Theorem~\ref{thm:approx-higher-moments}.

\subsection{Geometric Isolation for Lower Bounds}
\label{sec:overview-index}

The majority of our lower bounds reduce from the one-way communication problem \Index$_N$. Alice receives a string $z\in\{0,1\}^N$, Bob receives an index $q\in[N]$, and, after receiving one message from Alice, Bob must recover $z_q$. Any constant-error protocol requires $\Omega(N)$ bits of communication.  In this overview, we use $q$ for Bob's queried index; in the later formal reductions, the corresponding index may also be denoted by $j$ or a pair $(i,j)$ for different reductions.

Our reductions from \Index\ share the standard communication interface: Alice processes a stream prefix determined by $z$ and sends the resulting memory configuration to Bob, who then appends a suffix determined by $q$ and continues running the streaming algorithm. Once the full stream produces two well-separated values of the statistic, depending on whether $z_q=0$ or $z_q=1$, the communication lower bound for \Index\ gives the space lower bound.  The difficult part is constructing such a stream, and three challenges come up.
\begin{enumerate}
	\item[(i)] \emph{Global dependence on the suffix.} Every density is an inner product with the \emph{final} vector sum, so Bob's suffix changes the densities of \emph{all} earlier vectors, not only of the vectors encoding $z_q$ (the \emph{queried vectors}). The contribution of the queried vectors must therefore be separated multiplicatively from everything else, after accounting for how the suffix shifts every other density.
	
	\item[(ii)] \emph{Scarce dimension.} Any lower bound superlinear in $d$ must encode $N\gg d$ bits using vectors in $\R^d$. Orthogonal directions run out after $d$ bits, so packing more requires a different encoding: for the diversity index, each vector stores an entire row of a $\Theta(d) \times \Theta(d)$ bit matrix in a small perturbation, so $\Theta(d)$ vectors carry $\Theta(d^2)$ bits; for higher moments, we use low-correlation codes, whose distinct codewords have small nonzero inner products, so every nonqueried bit leaks into $M_p$ and the interference must be bounded.
	
	\item[(iii)] \emph{Geometric realizability.} All stream vectors must be unit vectors with bounded-precision coordinates. Under nonnegative cosine, every
	pairwise inner product must also be nonnegative. This constraint rules out the cancellation that lower-bound constructions usually rely on, such as pairing a
	vector with its negation to simulate a deletion in an insertion-only stream.
\end{enumerate}

We begin with the easier reductions as warm-ups. For $M_1$ and $M_2$, Alice encodes one bit per dimension with orthogonal basis vectors, inserting $e_k$ if
$z_k=1$ and $f_k$ otherwise. For exact $M_1$, Bob appends $e_q$, which shifts $M_1$ by an additive constant depending on $z_q$: enough for exact computation, but not for approximation. For approximating $M_2$, the queried bit must instead
create a multiplicative gap. Rather than appending copies of $e_q$, which would have large densities among themselves, Bob appends $L := N = \Theta(d)$ orthonormal
vectors, each with inner product $1/\sqrt{L}$ with $e_q$ and orthogonal to every other $e_k$ and every $f_k$, so that their sum is $\sqrt{L}\,e_q$. If $z_q=1$, the density of $e_q$ rises to $1+\sqrt{L}$, while each suffix vector keeps density only $1+1/\sqrt{L}$; hence $M_2\approx 3L$, versus $M_2=2L$ when $z_q=0$. A Hadamard construction realizes such a suffix while keeping all pairwise inner products nonnegative (Theorem~\ref{thm:second-density-moment}).

For the diversity index $M_{-1}$, setting $m := \Theta(d)$ and $N := m^2$, Alice interprets her bit string $z$ as an $m \times m$ Boolean matrix and encodes each row as one vector in $\mathbb{R}^{2m+1}$: a basis vector plus a
small perturbation recording the row's bits. Write Bob's index $q$ as $(i, j) \in [m]^2$, and call the $i$-th row the {\em queried row} and the $j$-th column the {\em queried column}. In the first block, Bob appends many copies of a vector that overlaps with every row except the queried row. This floods all the other rows, so they contribute almost nothing. In the second block, he appends many copies of a vector that tests the queried row's bit in the queried column. If the queried bit is one, the perturbation makes the queried row similar to these test vectors, so it gets flooded as well and contributes almost nothing. If the bit is zero, the row stays isolated and contributes about one. Each of Bob’s two blocks of copies contributes about one as well. Thus, the diversity index is about $3$ when the bit is zero and about $2$ when it is one, giving a constant-factor gap and hence the $\Omega(d^2)$ lower bound (Theorem~\ref{thm:quadratic-lower-bound}). 

For the higher-moment lower bounds, we address challenge~(ii) using code families that provide many query directions in low dimension. Since these directions are not orthogonal, the contributions from nonqueried codewords must be carefully controlled. In the signed setting, challenge~(i) can be handled by cancellation. Alice inserts {\em antipodal} pairs, making the sum of her prefix zero, so the densities are determined entirely by Bob's suffix. This allows the queried codewords to become dominant by a multiplicative factor, and a membership-query construction together with a standard counting argument gives the nearly tight lower bound. Under nonnegative cosine, however, challenge~(iii) rules out this approach: the nonnegativity constraint forbids antipodal pairs, so Alice's prefix necessarily contributes a positive baseline to the densities, potentially dominating the signal from the queried bit. To overcome this, we introduce \emph{equal-sum moment isolation}. Each bit is encoded by one of two insertion-only gadgets with the same vector sum, making the baseline independent of the encoded bit rather than forcing it to zero. A finite-difference identity then cancels part of this baseline shared by all vectors of a gadget (the \emph{local baseline}), so that the $p$-th moment contributions of the two gadgets differ only by a multiple of the $p$-th power of the final
vector sum's component along the gadget's codeword; Bob's query makes this component large for the queried bit. In this way, we recover the multiplicative separation needed for challenge~(i) while maintaining the nonnegativity requirement in~(iii). The remainder of this overview develops the construction, starting with cancellation in the signed setting and then turning to the gadget for the nonnegative setting, its geometric realization, and the resulting reduction from \Index.

\subsection{Equal-Sum Moment Isolation: A Substitute for Cancellation}
\label{sec:technical-highlight}

\vspace{2mm}\noindent{\bf Cancellation in Lower Bound Arguments.\ }
A common turnstile lower bound argument lets Bob use deletions as a \emph{late selector}. Alice inserts a large collection containing many possible answers; after learning his query, Bob deletes the part that is common to his side information, leaving a small residual that reveals the answer. For example, strict-turnstile lower bounds for sampling based on {\em Universal Relation} let Alice insert a support $X$ and Bob delete a known subset $Y\subseteq X$, so that a sampler sees the residual
$X\setminus Y$~\cite{KNP+17}. The bounded-deletion framework makes this dependence on cancellation explicit~\cite{JW18}. However, a cancellation-based turnstile lower bound does not automatically imply the same lower bound for insertion-only algorithms.

Although our stream is insertion-only, signed cosine provides a geometric analogue of deletion: inserting two antipodal vectors causes them to cancel in the vector sum; this is the geometric counterpart of cancellation in a turnstile reduction. Consider the lower bound for even moments under signed cosine
(Theorem~\ref{thm:signed-lower-bound}). Alice selects codewords and, for each selected unit vector $a$, inserts the antipodal pair $(a, -a)$. Her prefix therefore has vector sum zero, and thus the densities of her vectors are determined entirely by Bob’s suffix. Bob appends a suffix whose vector sum is a positive multiple of the
codeword $x$ he queries. The two vectors in an antipodal pair then have opposite densities, and for even $p$ their combined contribution is
$$
\ip{a}{s}^p+\ip{-a}{s}^p = 2\abs{\ip{a}{s}}^p,
$$
where the vector sum $s$ is proportional to $x$. A self-correlation is large, whereas all cross-correlations are small. The final moment therefore reveals whether the queried codeword was selected.

More precisely, Bob appends $m := \Theta(d)$ mutually orthogonal vectors whose sum is $s = \sqrt{m} x$. If $T$ is Alice's set of selected unit codewords, then the full stream satisfies the identity
$$
M_p=2m^{p/2}\sum_{a\in T}\abs{\ip{a}{x}}^p + m.
$$
The final term $m$ is the contribution of Bob's suffix. If $x\in T$, the self-correlation contributes at least $2m^{p/2}$; if $x\notin T$, the low-correlation code makes the first term small. Thus, the queried signal has scale $m^{p/2}$, while Bob's suffix contributes only $m$, an amplification of $m^{p/2-1}$. This growing ratio is exactly why the construction gives the nearly matching lower bound for even $p>2$.

\vspace{2mm}\noindent{\bf The Challenge under Nonnegative Cosine.\ }
Suppose $u_1,\ldots,u_r$ are unit vectors with pairwise nonnegative inner
products. Then
$$
\Big\|\sum_{i\in[r]} u_i\Big\|_2^2=\sum_{i\in[r]}\norm{u_i}_2^2 + 2\sum_{i<j}\ip{u_i}{u_j} \geq r.
$$
Thus, a nonempty prefix cannot have zero vector sum. Alice necessarily creates a positive baseline that cannot be canceled when Bob appends his query.

\vspace{2mm}\noindent{\bf Equal-Sum Gadgets.\ }
Our key new idea is to replace explicit cancellation with \emph{equal-sum moment isolation}.
For each codeword $a_i$, Alice has two possible multisets,
$\Gzero(a_i)$ and $\Gone(a_i)$, representing the bits zero and one. They
satisfy
$$
\abs{\Gzero(a_i)}=\abs{\Gone(a_i)}\ ,
\quad\text{and}\quad
\sum_{u\in\Gzero(a_i)}u = \sum_{u\in\Gone(a_i)}u.
$$
The common sum is not zero. Rather, the key point is that it is the
same for the two encodings. Since every density is an inner product with the final vector sum, replacing $\Gzero(a_i)$ by $\Gone(a_i)$ does not change the density of any vector outside the $i$-th gadget. The swap can therefore affect $M_p$ only through the different internal contributions of the two multisets.

This is the first part of equal-sum moment isolation: rather than forcing Alice's sum to zero, we make it independent of her bit string. The second part is to make the effect of swapping a $0$-gadget for a $1$-gadget take a clean form. Writing $Z_i := \langle a_i, s \rangle$ for the component of the final vector sum along the $i$-th codeword, we arrange that this swap changes
$M_p$ by a quantity proportional to $Z_i^p$, independent of the local baseline shared by the gadget's vectors. Besides Bob's query, $Z_i$ also receives a
contribution from Alice's prefix; this part, the \emph{code-space baseline}, is {\em not} cancelled and is controlled by the choice of code (see below). The parameters are chosen so that, for Bob's queried index $q$, $Z_q^p$ is on the same scale as the \emph{global baseline} (the value of $M_p$ on the all-zero instance), while the combined $p$-th power responses of all nonqueried gadgets are only a small fraction of that scale. Consequently, changing the queried gadget produces a constant-order perturbation of $M_p$, whereas the changes contributed by all nonqueried gadgets remain controlled. We next show how a finite-difference construction produces this clean $Z_i^p$ dependence.

\vspace{2mm}\noindent{\bf A Finite-Difference Construction.\ }
In our hard instance, the vectors associated with the $i$-th bit (i.e., $\Gzero(a_i)$ or $\Gone(a_i)$) come at several \emph{levels}. A level-$j$ vector has density
$B_i+t_jZ_i.$ Here $B_i$ is the \emph{local baseline}, which is query-independent, the same for all vectors of the gadget, and contributed by Alice's prefix and the local gadget structure, while $Z_i = \langle a_i, s \rangle$ is as defined above. The parameter $t_j$ determines how strongly level $j$ amplifies this response. If we simply take a $p$-th
power, then
$$
(B_i+t_jZ_i)^p = \sum_{r=0}^p\binom pr B_i^{p-r}t_j^rZ_i^r.
$$
The desired term is $Z_i^p$, but it is accompanied by lower-order mixed terms involving the local baseline $B_i$. To obtain a clean bit-dependent contribution, we need to eliminate these mixed terms when comparing the two gadgets. The purpose of the finite-difference gadget is to remove all of these mixed terms without deleting any vector and without violating the nonnegative-cosine promise.

Let $c_j=(-1)^{p-j}\binom pj$ and $t_j=t_0+jh$ for $j=0,1,\ldots,p$ and some constants $t_0, h > 0$. 
The coefficients $c_j$ define the $p$-th forward-difference operator. We have
$$
\sum_{j=0}^p c_jt_j^r=0
\quad \forall 0\leq r<p, \qquad\text{whereas}\qquad
\sum_{j=0}^p c_jt_j^p=p!\,h^p.
$$
Applying these identities after expanding the $p$-th power gives, for
all $B,Z\in\R$,
$$
\sum_{j=0}^p c_j(B+t_jZ)^p = p!\,h^pZ^p.
$$
Every term containing the local baseline $B$ is eliminated.

Of course, the signs of the $c_j$ cannot be implemented as negative multiplicities in an insertion-only stream. We instead split the finite difference into two multisets. Levels with $c_j>0$ appear in $\Gone$, levels with $c_j<0$ appear in $\Gzero$, and every level is inserted with the positive multiplicity $\abs{c_j}$. The negative signs exist only in the mathematical comparison ``one-gadget minus zero-gadget''; neither stream contains a negative update. Since each level will be realized by two vectors with the same density (one for each balancing sign; see below), the difference between the contributions of the two gadgets is exactly
\begin{equation}
	\label{eq:kappa-p}
	\sum_{u\in\Gone(a_i)}\ip{u}{s}^p - \sum_{u\in\Gzero(a_i)}\ip{u}{s}^p = \kappa_p Z_i^p,
	\qquad \text{where} \quad
	\kappa_p=2p!h^p>0.
\end{equation}

\vspace{2mm}\noindent{\bf Local Anchors, Balancing Components, and Colors.\ }
We next realize these algebraic levels by unit vectors whose
pairwise inner products are all nonnegative.  Set $m := \Theta(d)$. We use three mutually orthogonal blocks. The first is the \emph{code block}, containing a sparse nonnegative unit codeword $a_i\in\mathbb{R}_{\geq 0}^m$ for each bit $z_i$ of Alice’s input string. The other two blocks have orthonormal bases
$e_1,\ldots,e_m$ and $g_1,\ldots,g_m$. A balanced coloring
$\chi:[N]\to[m]$ assigns each codeword to one of these $m$ pairs of
auxiliary directions. The level-$j$ vectors for bit $i$ are
$$
u_{i,j,\sigma} = t_ja_i+\lambda e_{\chi(i)}+\sigma\mu_jg_{\chi(i)}, \qquad \sigma\in\{-1,+1\},
$$
where $\mu_j=\sqrt{1-\lambda^2-t_j^2}$ makes the vector have unit norm.

The coordinate $\lambda e_{\chi(i)}$ is the \emph{local anchor}. It is local because only vectors of the same color share it. The coordinate $\sigma\mu_jg_{\chi(i)}$ is the \emph{balancing component}. Both signs are inserted at every level, so the balancing components cancel from the vector sum of a gadget. They nevertheless serve an important geometric purpose: their level-dependent lengths allow all level vectors to be normalized without changing the coefficients $t_j$ in the code block.

The anchor protects nonnegativity. Vectors of different colors have orthogonal anchor and balancing components, so their inner product is simply $t_jt_{j'}\ip{a_i}{a_{i'}}\geq0$. For vectors of the same color, opposite balancing signs may contribute a negative term, but the common anchor contributes $\lambda^2$. The constants are chosen so that the anchor always dominates the worst possible negative balancing contribution. Consequently, every pair of gadget vectors has nonnegative inner product.

The coloring is needed for dimension efficiency. Giving every one of the $N$ encoded bits its own anchor and balancing directions would require dimension $\Omega(N)$, which would destroy the superlinear lower bound. We instead reuse only $m$ anchor--balancing pairs. A balanced coloring assigns $N/m$ gadgets to each color, which keeps the local baseline $B_i$ uniform and of order $N/m$. Thus, the same construction simultaneously preserves nonnegative geometry, uses only $O(m)$ auxiliary dimensions, and keeps the local baseline under control: although $B_i$ cancels within each gadget, it still contributes to the global baseline.

Pairing the two balancing signs at each level gives
$$
u_{i,j,+}+u_{i,j,-} = 2t_j a_i+2\lambda e_{\chi(i)},
$$
so the balancing components cancel from the gadget sum, and what remains is affine in $t_j$. The cases $r=0$ and $r=1$ of the identity $\sum_j c_j t_j^{\,r}=0$ then imply that the $0$- and $1$-gadgets have the same cardinality and vector sum. Thus switching between the two gadgets does not change the contribution of this bit to Alice's prefix sum.

Moreover, because the two balancing signs occur with equal multiplicity, Alice's prefix has no component in any $g$-direction, and Bob's suffix introduced below also has no such component. Hence the final stream sum $s$ has no balancing component. For a level vector
$$
u_{i,j,\sigma} = t_j a_i+\lambda e_{\chi(i)}+\sigma\mu_jg_{\chi(i)},
$$
its density is therefore
$$
\langle u_{i,j,\sigma},s\rangle = \lambda\langle e_{\chi(i)},s\rangle
+t_j\langle a_i,s\rangle.
$$
Writing $B_i:=\lambda\langle e_{\chi(i)},s\rangle$ and recalling $Z_i:=\langle a_i,s\rangle$, we obtain the affine form $B_i+t_jZ_i$, exactly the dependence on the level parameter $t_j$ needed for the finite-difference identity.

\vspace{2mm}\noindent{\bf The Query and Sparse Nonnegative Codes.\ }
Alice receives $z\in\{0,1\}^N$ and inserts a gadget (that is, a multiset of vectors) $\mathcal G_{z_i}(a_i)$ for each bit $z_i$ and codeword $a_i$.
By the equal-sum property, the vector sum of this entire prefix is
independent of $z$. Bob receives a query index $q$. In a fourth,
orthogonal block with basis $w_1,\ldots,w_m$, he uses the unit vectors
$$
v_\ell = \frac{1}{\sqrt m}a_q + \sqrt{1-\frac1m}\,w_\ell, \qquad \ell\in[m],
$$
and appends $L=m^{(p-2)/2}$ copies of each. Their total code-space component is $Qa_q$, where $Q=L\sqrt m=m^{(p-1)/2}.$ This choice of $L$ balances the query signal with the contribution of Bob's suffix itself: there are $mL$ suffix vectors, each with density $O_p(L)$, and hence $mL^{p+1}=O_p(Q^p).$ The suffix vectors have nonnegative inner products with one another and
with every prefix vector.

The final code response seen by gadget $i$ is
$$
Z_i = \gamma_p\sum_{r\in[N]}\ip{a_i}{a_r} + Q\ip{a_i}{a_q},
$$
where $\gamma_p>0$ is a constant depending only on $p$. The first term is the code-space baseline created by Alice's prefix, while the second is Bob's query response. Unlike the local baseline $B_i$, the first term is not removed by the finite difference, because it enters through $Z_i$; bounding its total contribution is what limits $N$ below. In particular, the self-correlation $\ip{a_q}{a_q}=1$ gives $Z_q\geq Q$.

The codewords are normalized incidence vectors of sparse subsets of $[m]$. Besides being nonnegative, they satisfy two correlation bounds (Lemma~\ref{lem:code}). First, for a sufficiently small constant $\iota$,
$$
\sum_{i\neq q}\ip{a_i}{a_q}^p\leq\iota,
$$
which limits the total leakage of Bob's query into nonqueried gadgets.
Second, writing
$
R_{\max} = \left(N/m\right)^{(p-1)/(p-2)},
$
their first power correlation row sums satisfy
$$
1+\sum_{r\neq i}\ip{a_i}{a_r} = O_p(R_{\max}).
$$
This controls the code-space baseline contributed by Alice's prefix.
Together, the two bounds imply
$$
\sum_{i\neq q}Z_i^p = O_p\bigl(NR_{\max}^p+\iota Q^p\bigr).
$$

\vspace{2mm}\noindent{\bf The Final Reduction.\ }
Let $V_0(q)$ be the value of $M_p$ on the all-zero instance, in which
Alice uses $\Gzero(a_i)$ at every index and Bob asks query $q$. Because the two bit-gadgets have the same vector sum, the final stream sum is independent of Alice's bit string; consequently, so are the responses $Z_i$ and the densities of Bob's suffix vectors. We may therefore change Alice's gadgets one at a time while keeping all other density parameters fixed, and apply the finite-difference identity separately to each changed gadget, giving the exact decomposition
$$
M_p(P(z,q)) = V_0(q)+\kappa_p\sum_{i:z_i=1}Z_i^p,
$$
where $\kappa_p$ is defined in \eqref{eq:kappa-p}.

We set $N:=m\lfloor c_Nm^{\rho_p-1}\rfloor=\Theta_p(m^{\rho_p})$, where $c_N>0$ is a sufficiently small constant,
and $\rho_p=\frac{p^2(p-1)}{2(p^2-2)}.$
The exponent $\rho_p$ is obtained by matching the polynomial scales of
Alice's aggregate code-space baseline and Bob's signal:
$$
NR_{\max}^p = N\left(\frac Nm\right)^{p(p-1)/(p-2)} = \Theta_p(Q^p).
$$
After fixing this exponent $\rho_p$, we first choose $\iota$ small enough to
control the query-leakage term $\iota Q^p$, and then choose the constant
$c_N$ small enough to make $NR_{\max}^p$ a small multiple of $Q^p$.
Thus, for any desired constant $\nu>0$, we have
$$
\sum_{i\neq q}Z_i^p\leq \nu Q^p.
$$

The global baseline $V_0(q)$ contains several common contributions. The local-baseline contribution is
$
O_p\big(N\left(N/m\right)^p\big)=o(Q^p);
$
Alice's code-space baseline contributes $O_p(NR_{\max}^p) = O_p(Q^p)$; Bob's query contributes $O_p((1+\iota)Q^p)$ across the zero-gadgets, where the $1$
accounts for the queried codeword itself; and Bob's suffix contributes
$O_p(mL^{p+1})=O_p(Q^p)$. Consequently, $V_0(q)=O_p(Q^p)$.

If $z_q=1$, the queried gadget contributes at least $\kappa_pQ^p$. If $z_q=0$, the total contribution of all selected nonqueried gadgets can be made at most $\kappa_p\nu Q^p$ for a small constant $\nu$. Since the global baseline is only a constant multiple of $Q^p$, these bounds give a constant multiplicative gap, allowing an approximation to $M_p$ to
recover $z_q$.

Finally, the construction has dimension $\Theta(m)$ and encodes
$N=\Theta_p(m^{\rho_p})$ bits, so the one-way lower bound for
\Index$_N$ gives $\Omega_p(d^{\rho_p})$ bits (Theorem~\ref{thm:main}). A small common perturbation, followed by rational approximation, makes every vector use only $O_p(\log d)$ bits per coordinate while preserving the constant gap and nonnegative pairwise cosines.

\vspace{2mm}\noindent{\bf Discussion.\ }
Our construction combines two ideas. First, the equal-sum property makes the final vector sum independent of Alice's choice between the $0$- and $1$-gadgets. Switching gadgets therefore changes no density outside the gadget, so the effect of each bit stays local. Second, within a gadget, finite differences exploit the polynomial structure of the $p$-th moment, so that the contributions of the two encodings to $M_p$ differ by exactly $\kappa_p Z_i^p$, independent of the local baseline. Finally, the sparse code and Bob's suffix make the queried term $Z_q^p$ comparable to the global baseline $V_0(q)$, while all nonqueried terms together contribute only a small fraction of this scale. This gives the constant multiplicative gap needed to recover the queried bit.

The equal-sum idea may be useful more broadly in insertion-only lower bounds. The principle is to construct two legal gadgets that encode different bits but agree on everything through which they interact with the rest of the instance (here, just the vector sum) so that switching between them affects nothing else. How to then separate the two local contributions is problem-specific: for density moments, finite differences do this, but they rely on the polynomial structure above and may not extend directly to other objectives.

\section{Preliminaries}
\label{sec:preliminaries}

\noindent{\bf Computation Model and Problem Setup.\ }
In the data stream model, the input is a sequence of nonzero vectors $P=(x_1,\ldots,x_n)$, where each $x_i\in\mathbb{R}^d$. The algorithm is allowed to make a sequential pass over the stream while using limited memory. At the end of the stream, it computes the target function on $P$ using its final memory configuration.

For each $i \in [n]$, let $u_i:={x_i}/{\|x_i\|_2}$ be the normalized vector. We slightly abuse the notation and also write $P=(u_1,\ldots,u_n)$ as the input stream, since each input vector can be normalized upon arrival.
For each prefix $P_t=(u_1,\ldots,u_t)$, write
$s_t:=\sum_{j\leq t}u_j$ and define the prefix density of $u_i$ by
$$
D_i(t):=\sum_{j\leq t}\langle u_i,u_j\rangle = \langle u_i,s_t\rangle
$$
for any $i\leq t$.
We write
$
M_p(P_t):=\sum_{i\leq t}D_i(t)^p
$, 
and $D_i:=D_i(n)$.

We distinguish two regimes. In the \emph{nonnegative-cosine regime}, we assume that $\langle u_i,u_j\rangle\geq 0$ for every pair
$i,j\in[n]$. Consequently, every prefix density satisfies
$D_i(t)\geq 1$ and is nondecreasing in $t$. In the \emph{signed-cosine regime}, no sign restriction is
imposed. The densities may then be zero or negative, so inverse moments need not be defined and odd higher moments may be negative or suffer from severe cancellation. We therefore focus in this regime on $p = 1$ and on even integers $p$, for which $M_p(P)$ is a nonnegative statistic.

\vspace{2mm}
\noindent{\bf Conventions.\ }
Unless explicitly stated otherwise, $p$ denotes an integer.
We focus on integer exponents $p \ge -1$. By convention, we define \(M_0(P):=n\), the stream length. Some results under nonnegative cosine extend to arbitrary fixed real exponents $p>2$, and we state these extensions explicitly.

Our upper bounds are measured in words, where one word stores an input coordinate or an intermediate real-valued scalar; thus, a vector in $\mathbb{R}^d$ occupies $O(d)$ words. Our lower bounds are stated in the bit model. The hard instances use rational coordinates representable with $O_p(\log d)$ bits each, and the algorithm's memory is measured in bits. This finite-precision requirement is necessary for information-theoretic lower bounds, since an unrestricted real number could encode arbitrarily many bits in a single coordinate or memory word.

For simplicity, we describe our algorithms using fully random hash functions and random signs. These can be replaced by standard bounded-independence constructions or pseudorandom generators for space-bounded computation, incurring only polylogarithmic overhead and a negligible additional error probability.

We use the terms ``point'' and ``vector'' interchangeably.  Unless otherwise specified, $\norm{\cdot}$ denotes the Euclidean norm. We omit floors and ceilings when they do not affect the correctness of the analysis. All logarithms are base $2$.  For $X \ge 0$, we say that $\widetilde{X}$ is a $(1+\eps)$-approximation to $X$ if
$(1-\eps)X \leq \widetilde{X} \leq (1+\eps)X.$

\vspace{2mm}
\noindent{\bf Communication Complexity.\ }
To establish our space lower bounds, we reduce from the classical communication complexity problem \Index.  In the \Index$_N$ problem, Alice receives a
string $X\in\{0,1\}^N$, Bob receives an index $J\in[N]$, and Bob must output
$X_J$ after receiving one message from Alice.

\begin{lemma}[see, e.g., \cite{KN96}]
	\label{lem:index-lower-bound}
	Every one-way protocol for \Index$_N$ with error at most
	$1/3$ uses $\Omega(N)$ bits of communication.
\end{lemma}

\vspace{2mm}
\noindent{\bf Roadmap.\ }
The rest of the paper is organized as follows. Section~\ref{sec:basics} collects the algebraic and spectral identities that underlie our exact algorithms. Section~\ref{sec:diversity} studies the diversity index $M_{-1}$ in the nonnegative-cosine regime, giving a one-pass approximation algorithm and a matching lower bound. Section~\ref{sec:signed} considers even density moments under signed cosine, where tensorization and a lifted AMS sketch give one-pass approximation algorithms and nearly matching lower bounds. Section~\ref{sec:nonnegative} gives one-pass upper bounds for every fixed real $p > 2$ under nonnegative cosine, and Section~\ref{sec:nonnegative-lb} proves the corresponding lower bounds via equal-sum moment isolation. We conclude in Section~\ref{sec:conclusion}.

\section{Algebraic Identities and Exact Computation}
\label{sec:basics}

We begin with three simple mathematical facts.  The first two are algebraic consequences of the finite-dimensional cosine feature map; they directly give the algorithms for the exact computation of $M_p$ for $p \ge 1$. The third is a spectral consequence of positive semidefiniteness and
nonnegativity.

For a multi-index vector $\vec{\alpha}=(\alpha_1,\ldots,\alpha_d)\in\mathbb{N}^d$ and a vector $z \in \mathbb{R}^d$, write $|\vec{\alpha}|:=\sum_{k \in [d]}\alpha_k$,  
$z^{\vec{\alpha}}:=\prod_{k \in [d]} z_k^{\alpha_k}$, and multinomial coefficient
$
\binom{p}{\vec{\alpha}}:=\frac{p!}{\alpha_1!\cdots\alpha_d!}.
$

\subsection{Exact Computation of $M_1$}
\label{sec:SS}

The first density moment collapses to the squared norm of the  vector
sum.  Recall that $s_n=\sum_{i\in[n]}u_i$ denotes the sum of the normalized
vectors in the stream.

\begin{theorem}
	\label{thm:SS}
	We have $M_1(P)=\norm{s_n}_2^2.$ Consequently, $M_1(P)$ can be computed exactly in one pass using $O(d)$ words.
\end{theorem}

\begin{proof}
	By bilinearity,
	$$
	M_1(P) = \sum_{i \in [n]}\sum_{j \in [n]}\ip{u_i}{u_j} = \ip{\sum_{i \in [n]} u_i}{\sum_{j \in [n]} u_j} = \norm{s_n}_2^2.
	$$
	The algorithm maintains $s_t=s_{t-1}+u_t$ and outputs its squared Euclidean
	norm at the end.
\end{proof}

We complement the upper bound with a lower bound matching its linear dependence on $d$.

\begin{theorem}
	Any one-pass streaming algorithm that computes $M_1(P)$ exactly with
	probability at least $2/3$ on every stream of vectors in $\mathbb{R}^d$
	requires $\Omega(d)$ bits of space. The lower bound holds for
	streams of length $O(d)$ consisting of standard basis vectors, and hence
	in the nonnegative-cosine regime.
\end{theorem}

\begin{proof}
	We reduce from $\Index_m$, where Alice holds $z \in \{0,1\}^m$ and Bob holds an index $j\in[m]$.

	Let $m:=\left\lfloor {d}/{2}\right\rfloor,$ and let $e_1,\ldots,e_m,f_1,\ldots,f_m$
	be mutually orthonormal vectors in $\mathbb{R}^d$. For each $i\in[m]$,
	Alice inserts the vector
	$$
	a_i :=
	\begin{cases}
		e_i, & \text{if } z_i=1,\\
		f_i, & \text{if } z_i=0.
	\end{cases}
	$$
	The vectors $a_1,\ldots,a_m$ are mutually orthogonal. Alice runs the
	streaming algorithm on this prefix and sends its memory configuration to Bob.
	
	Given $j\in[m]$, Bob resumes the algorithm and appends the vector $e_j$.
	We have
	$$
	M_1(P) = \Big\| \sum_{i \in [m]} a_i+e_j \Big\|^2.
	$$
	If $z_j=0$, then $a_j=f_j$ is orthogonal to $e_j$, and all $m+1$
	vectors in the completed stream are mutually orthogonal. Therefore,
	$M_1(P)=m+1.$ If $z_j=1$, then $a_j=e_j$, so the coefficient of $e_j$ in the final
	vector sum is two, while the remaining $m-1$ vectors are orthogonal to
	$e_j$ and to one another. Hence,
	$
	M_1(P)=2^2+(m-1)=m+3.
	$
	Thus, whenever the streaming algorithm outputs the exact value of
	$M_1(P)$, Bob can recover $z_j$ by distinguishing $m+1$ from $m+3$.
	
	If the streaming algorithm uses $s$ bits of memory, this construction
	gives a one-way protocol for $\Index_m$ communicating exactly
	$s$ bits and having error probability at most $1/3$. Consequently, by Lemma~\ref{lem:index-lower-bound},
	$s=\Omega(m)=\Omega(d).$ All vectors in the construction are standard basis vectors, so every pairwise cosine similarity belongs to $\{0,1\}$. The stream therefore
	lies in the nonnegative-cosine regime, has length $m+1=O(d)$, and uses
	coordinates with $O(1)$-bit complexity.
\end{proof}

\subsection{Exact Density Moments}
\label{sec:exact-Fp}

For an integer $p\geq1$, define the degree-$p$ counters
$A_{\vec{\alpha}}:=\sum_{i \in [n]} u_i^{\vec{\alpha}}$ for every $\vec{\alpha}\in\mathbb{N}^d\text{ with }|\vec{\alpha}|=p$.
These are the coordinates of the symmetric tensor
$\sum_{i \in [n]} u_i^{\otimes p}$ in a monomial basis.

\begin{theorem}
	\label{thm:higher-moments}
	For every integer $p\geq1$, $M_p(P)$ can be computed exactly in one pass using $O_p(d^p)$ words.
\end{theorem}

\begin{proof}
	For every $i$, we have
	$$
	D_i^p = \ip{u_i}{s_n}^p = \left(\sum_{k \in [d]} u_{i,k}s_{n,k}\right)^p = \sum_{|\vec{\alpha}|=p}\binom{p}{\vec{\alpha}}u_i^{\vec{\alpha}}s_n^{\vec{\alpha}}.
	$$
	Summing over $i$ and interchanging the two sums gives
	$$
	M_p(P) = \sum_{|\vec{\alpha}|=p}\binom{p}{\vec{\alpha}} \left(\sum_{i \in [n]} u_i^{\vec{\alpha}}\right)s_n^{\vec{\alpha}}.
	$$
	During the stream, the algorithm maintains
	$s_t$ and counter $A_{\vec{\alpha}}(t)=\sum_{i\leq t}u_i^{\vec{\alpha}}$ for every $\vec{\alpha}$ with $|\vec{\alpha}|=p$.
	There are $\binom{d+p-1}{p} = O_p(d^p)$ such counters.
\end{proof}

\subsection{The Diversity--Rank Principle}
\label{sec:diversity-rank}

Both nonnegative-cosine algorithms, for the diversity index (Section~\ref{sec:diversity}) and for higher moments (Section~\ref{sec:nonnegative}), rest on a structural fact that holds for every prefix of the stream.  

Let $U_t\in\R^{d\times t}$ be the matrix whose columns are $u_1,\ldots,u_t$, where $\langle u_i,u_j\rangle\geq 0$ for every pair $(i, j) \in [t]^2$.
Let $G_t:=U_t^\top U_t$ be the $t\times t$ cosine Gram matrix.  Define the diagonal matrix of row sums 
$
\Delta_t:=\diag\bigl(D_1(t),\ldots,D_t(t)\bigr).
$
Because every $D_i(t) =\sum_{j \in [t]}\langle u_i,u_j\rangle \geq1$, the matrix $\Delta_t$ is invertible.

\begin{lemma}
	\label{lem:rank-bound}
	In the nonnegative-cosine regime, every nonempty prefix $P_t = (u_1, \ldots, u_t)$ satisfies
	$$
	1\leq M_{-1}(P_t)\leq \rank(G_t)\leq d.
	$$
\end{lemma}

\begin{proof}
	Consider the symmetric normalization
	$
	B_t:=\Delta_t^{-1/2}G_t\Delta_t^{-1/2}.
	$
	Since $G_t\succeq0$, we have $B_t\succeq0$. Moreover, because $\Delta_t^{-1/2}$ is invertible, $\rank(B_t)= \rank(G_t)\leq d$.  We also have
	\begin{equation}
		\label{eq:trace-equals-di}
		\tr(B_t) = \sum_{i=1}^t\frac{(G_t)_{ii}}{D_i(t)} = \sum_{i=1}^t\frac{1}{D_i(t)} = M_{-1}(P_t),
	\end{equation}
	where we used $(G_t)_{ii}=\norm{u_i}_2^2=1$.
	
	Now define
	$
	H_t:=\Delta_t^{-1}G_t.
	$
	The entries of $H_t$ are nonnegative, and every row sums to one.  Thus $H_t$
	is row-stochastic and every eigenvalue of $H_t$ has absolute value at most one.
	The matrices $H_t$ and $B_t$ are similar because
	$
	B_t=\Delta_t^{1/2}H_t\Delta_t^{-1/2}.
	$
	Since $B_t$ is positive semidefinite, all of its eigenvalues are real and
	nonnegative.  Therefore, every eigenvalue of $B_t$ lies in $[0,1]$.
	Furthermore,
	$$
	B_t\Delta_t^{1/2}\1 = \Delta_t^{-1/2}G_t\1 = \Delta_t^{-1/2}\Delta_t\1 = \Delta_t^{1/2}\1,
	$$
	so $1$ is an eigenvalue of $B_t$.  It follows that
	$$
	1\leq\tr(B_t)\leq\rank(B_t)=\rank(G_t)\leq d.
	$$
	Combining this with~\eqref{eq:trace-equals-di} proves the lemma.
\end{proof}

\section{Diversity Index}
\label{sec:diversity}

In this section, we study the diversity index $M_{-1}$. Because the diversity index involves reciprocal densities, we work throughout this section in the nonnegative-cosine regime, where every density is positive. We first use the diversity--rank principle from Lemma~\ref{lem:rank-bound} to design a one-pass approximation algorithm, and then complement the upper bound with a lower bound that matches the dependence on the dimension $d$.

\subsection{The Algorithm}
\label{sec:diversity-one-pass}

Our one-pass algorithm for estimating $M_{-1}$ is described in Algorithm~\ref{alg:diversity}.

\begin{algorithm}[t]
	\caption{One-pass diversity-index estimator}
	\label{alg:diversity}
	\DontPrintSemicolon
	\KwIn{A stream $P = (x_1, \ldots, x_n)$ of points in $\mathbb{R}^d$ in the nonnegative-cosine regime, parameters $\eps, \delta\in(0,1)$, the stream length $n$.}
	\KwOut{A $(1+\eps)$-approximation to $M_{-1}(P)$}
	
	$\lambda\gets \frac{8}{3\eps^2}\ln\frac{4}{\delta}$, $N_{\max} = \lambda d + \sqrt{2 \lambda d \ln\frac{2n}{\delta}} + \frac{2}{3} \ln \frac{2n}{\delta}$,
	$s\gets \mathbf{0}\in\R^d$, and $R\gets\emptyset$\;
	
	\ForEach{incoming $x_t\ (t = 1, \ldots, n)$}{
		Normalize $u_t\gets x_t/\norm{x_t}_2$\;
		\ForEach{$(u_i,r_i,\widehat D_i)\in R$}{
			$\widehat D_i\gets\widehat D_i+\ip{u_i}{u_t}$\;
			\If{$r_i>\min\{1,\lambda/\widehat D_i\}$}{
				Delete $(u_i,r_i,\widehat D_i)$ from $R$ \tcp*{rejection sampling on stored points}
			}
		}
		$\widehat D_t\gets 1+\ip{u_t}{s}$\;
		Draw $r_t\sim\operatorname{Unif}(0,1)$\;
		\lIf{$r_t\leq\min\{1,\lambda/\widehat D_t\}$}{
			Insert $(u_t,r_t,\widehat D_t)$ into $R$
		}
		\lIf{$\abs{R} > N_{\max}$}{
			{\KwRet{\textnormal{\textsc{Overflow}}}}
		}
		$s\gets s+u_t$\;
	}
	\KwRet{$\displaystyle
		\widehat{M}_{-1}:=\sum_{(u_i,r_i,\widehat D_i)\in R}
		\frac{1}{\widehat D_i\min\{1,\lambda/\widehat D_i\}}$}  \tcp*{$\widehat{D}_i = D_i$ at the end}
\end{algorithm}

Fix a threshold parameter $\lambda>0$.  Give every point $u_i$ an independent random value $r_i\sim\operatorname{Unif}(0,1)$ and define its desired inclusion probability at time $t\geq i$ by
$
q_i(t):=\min\left\{1,\frac{\lambda}{D_i(t)}\right\}.
$
Because all future similarities are nonnegative, $D_i(t)$ is nondecreasing and
$q_i(t)$ is nonincreasing.  This monotonicity allows the algorithm to maintain
exactly the sample
$$
R_t:=\{i\leq t:r_i\leq q_i(t)\},
$$
since a point rejected once can never become eligible again. The final contribution of a stored point is defined as
\begin{equation*}
	\label{eq:simplified-contribution}
	\frac{1}{D_i\min\{1,\lambda/D_i\}} =
	\begin{cases}
		1/D_i, & D_i\leq\lambda,\\
		1/\lambda, & D_i>\lambda.
	\end{cases}
\end{equation*}

The indicator $I_i:=\mathbf{1}[i\in R_n]$ is an independent Bernoulli random variable with mean $q_i:=q_i(n)$, and the estimator of Algorithm~\ref{alg:diversity} can therefore be written as
$
\widehat{M}_{-1} = \sum_{i \in [n]} \left(I_i \cdot \frac{1}{q_iD_i}\right).
$

\begin{lemma}
	\label{lem:unbiased-variance}
	We have
	$
	\E\left[\widehat{M}_{-1}\right]=M_{-1}(P)
	$
	and
	$
	\Var\left[\widehat{M}_{-1}\right]\leq\frac{M_{-1}(P)}{\lambda}.
	$
\end{lemma}

\begin{proof}
	For every $i$,
	$
	\E\left[I_i\frac{1}{q_iD_i}\right]=\frac{1}{D_i},
	$
	which proves unbiasedness.  The summands are independent.  If $D_i\leq
	\lambda$, then $q_i=1$ and the $i$-th contribution is deterministic.  If
	$D_i>\lambda$, then $q_i=\lambda/D_i$ and
	$$
	\Var\left[I_i\frac{1}{q_iD_i}\right] = \frac{1-q_i}{q_iD_i^2} \leq \frac{1}{\lambda D_i}.
	$$
	Summing over $i$ gives the variance bound.
\end{proof}

\begin{theorem}
	\label{thm:diversity-main}
	Let $\eps, \delta\in(0,1)$. With probability at least $1 - \delta$, Algorithm~\ref{alg:diversity} computes a $(1+\eps)$-approximation to $M_{-1}(P)$ for any input stream $P$ of $n$ vectors in $\mathbb{R}^d$ in the nonnegative-cosine regime, using at most 
	$O\left(d^2\eps^{-2}\log\frac{1}{\delta}+d\log\frac{n}{\delta}\right)$ words of space.
\end{theorem}

\begin{proof}
	For every $i \in [n]$, let
	$
	Y_i:=I_i\frac{1}{q_iD_i}-\frac{1}{D_i}.
	$
	The variables $Y_i$ are independent and $\E[Y_i] = 0$.  For $D_i \le \lambda$, $Y_i = 0$.  For $D_i>\lambda$, the sampled contribution is either $0$ or $1/\lambda$, while $1/D_i<1/\lambda$.  Hence, $|Y_i|\leq1/\lambda$.  By Lemma~\ref{lem:unbiased-variance},
	$$
	\sum_i\E\left[Y_i^2\right] = \Var\left[\widehat{M}_{-1}\right] \leq \frac{M_{-1}(P)}{\lambda}.
	$$
	Bernstein's inequality (Lemma~\ref{lem:Bernstein}) gives
	\begin{eqnarray*}
		\Pr\left[ \Big|\sum_iY_i\Big|>\eps M_{-1}(P) \right]
		&\leq& 2\exp\left( -\frac{\eps^2M_{-1}(P)^2} {2M_{-1}(P)/\lambda+2\eps M_{-1}(P)/(3\lambda)} \right) \\
		&\leq& 2\exp\left(-\frac{3}{8}\lambda\eps^2M_{-1}(P)\right) \le \frac{\delta}{2},
	\end{eqnarray*}
	where the last inequality follows from the choice of $\lambda$ and the fact $M_{-1}(P)\geq1$ (by Lemma~\ref{lem:rank-bound}).  
	
	Let $N_t := \abs{R_t}$ be the number of tuples stored in $R_t$ at time $t$.  Then,
	\begin{equation}
		\label{eq:Nt}
		\E[N_t] = \sum_{i \in [t]} q_i(t) \leq\lambda\sum_{i \in [t]}\frac{1}{D_i(t)} = \lambda M_{-1}(P_t) \le \lambda d,
	\end{equation}
	where the last inequality follows from Lemma~\ref{lem:rank-bound}.
	
	For a fixed time $t$, the indicator variables $I_{i,t}:=\mathbf{1}[r_i \le q_i(t)]\ (i \in [t])$ are independent Bernoulli random variables.  Since $N_t = \abs{R_t} = \sum_{i \in [t]}I_{i, t}$, again by Bernstein’s inequality,
	$$
	\Pr\left[N_t > \E[N_t] + \sqrt{2\E[N_t] \ln\frac{2n}{\delta}}+\frac{2}{3} \cdot \ln\frac{2n}{\delta}\right] \leq \exp\left(-\ln\frac{2n}{\delta}\right) = \frac{\delta}{2n}.
	$$
	Combining this with \eqref{eq:Nt} and taking a union bound over all $n$ time steps, we have with probability at least $1 - \delta/2$, 
	$$
	N_t \le \lambda d + \sqrt{2 \lambda d \ln\frac{2n}{\delta}} + \frac{2}{3} \ln \frac{2n}{\delta} = O\left(d\eps^{-2}\log\frac{1}{\delta}+\log\frac{n}{\delta}\right)
	$$
	for all $t \in [n]$, in which case the algorithm will not return \textsc{Overflow}.
	
	The space bound follows by multiplying the space cap $N_{\max}$ by the vector dimension $d$.
\end{proof}

\begin{remark}[Unknown stream length]
	We note that if the stream length $n$ is not known in
	advance, the space cap $N_{\max}$ can be adjusted using a standard doubling method.
\end{remark}

\subsection{The Lower Bound}
\label{sec:diversity-lb}

We next show that the quadratic dependence on the dimension is inherent by another reduction from \Index$_N$.

\vspace{2mm}
\noindent{\bf Input Reduction.\ }
We choose $N:=m^2$, where $m := \lfloor(d-1)/2\rfloor$.   If $2m+1 < d$, we embed the construction into $\R^d$ by appending zero coordinates.

We interpret Alice's input in  \Index$_N$ as a matrix $X \in \{0,1\}^{m \times m}$, and Bob's {\em query index} as a pair $J = (i, j) \in [m]^2$.

Let $\{e_1,\ldots,e_m,f_1,\ldots,f_m,g\}$ be an orthonormal basis of $\R^{2m+1}$, and set parameters $\eta:=m^{-2}$, $\gamma:=m^{-4}$, and $L:=m^3$. Given $X=(X_{k\ell})\in\{0,1\}^{m\times m}$,  Alice creates the set of vectors
\begin{equation}
	\label{eq:lower-bound-alice-vector}
	a_k := e_k+\eta\sum_{\ell \in [m]} X_{k\ell}f_\ell+\gamma g,
	\qquad  k = 1, \ldots, m.
\end{equation}
Bob creates $L$ copies of each of the following two vectors
\begin{equation}
	\label{eq:lower-bound-bob-vectors}
	b_i := \sum_{k\neq i}e_k+\gamma g, \quad \text{and} \quad
	q_j := f_j+\gamma g.
\end{equation}
Intuitively, the {\em $e$-coordinates} identify rows and allow Bob to suppress all nonqueried rows relative to the queried row (i.e., the $i$-th row); the {\em $f$-coordinates} encode the matrix entries and allow Bob to test the queried bit $X_{ij}$; and the {\em $g$-coordinate} adds a small common component that makes every pairwise cosine similarity strictly positive.
The hard input stream consists of
\begin{equation}
	\label{eq:hard-stream}
	P(X; (i,j)) :=(a_1,\ldots,a_m, \underbrace{b_i,\ldots,b_i}_{L\text{ copies}}, \underbrace{q_j,\ldots,q_j}_{L\text{ copies}}).
\end{equation}

For every nonzero input vector $v$, write $\bar v=v/\norm{v}_2$ as the normalized vector.   The cosine similarities between input vectors are the inner products among the corresponding normalized vectors $\bar a_k$, $\bar b_i$, and $\bar q_j$.  Every input coordinate is a nonnegative rational with $O(\log m)$ bits.  Moreover, every two raw vectors have inner product at least $\gamma^2$, because they share the positive $g$-coordinate.  Hence, every pairwise cosine similarity
in~\eqref{eq:hard-stream} is strictly positive.

The construction has the following interpretation.  The vector $b_i$ is nearly
orthogonal to the queried row vector $a_i$, but has cosine similarity
$\Omega(m^{-1/2})$ with every other Alice vector.  Its $L$ copies therefore
make all rows except $i$ contribute negligibly to the diversity index.  The
vector $q_j$ tests the queried bit: its cosine similarity with $a_i$ is about
$\eta$ when $X_{ij}=1$ and only about $\gamma^2$ when $X_{ij}=0$.  Finally,
each of the two duplicate blocks contributes approximately one unit of
diversity.  The total is therefore close to $3$ when $X_{ij}=0$ and close to
$2$ when $X_{ij}=1$.

\begin{lemma}
	\label{lem:hard-instance-gap}
	For every $X\in\{0,1\}^{m\times m}$ and $(i,j)\in[m]^2$, the stream
	$P(X; (i,j))$ has $m+2m^3$ vectors in dimension $2m+1$.  Moreover,
	if $X_{ij}=0$, then
	$
	M_{-1}(P)\geq 3-\frac{4}{m^2}.
	$
	If $X_{ij}=1$, then
	$
	M_{-1}(P)\leq 2+\frac{4}{m}.
	$
\end{lemma}

\begin{proof}
	Let $A_k:=\norm{a_k}_2$, $B:=\norm{b_i}_2$, and $Q:=\norm{q_j}_2$.
	We have
	$$
	A_k^2 \leq 1 + m\eta^2 + \gamma^2 = 1+m^{-3}+m^{-8} \leq 2, \quad B^2 = 	m-1+\gamma^2 \leq m, \quad \text{and} \quad Q^2 = 1+\gamma^2\leq 2.
	$$
	
	We first bound the contribution of the queried vector $\bar a_i$.  If
	$X_{ij}=0$, then, for every $k\neq i$, we have
	$
	\ip{\bar a_i}{\bar a_k} \leq m\eta^2+\gamma^2,
	$
	$
	\ip{\bar a_i}{\bar b_i} \leq \gamma^2,
	$ and 
	$
	\ip{\bar a_i}{\bar q_j} \leq \gamma^2.
	$
	Its density therefore satisfies
	\begin{equation*}
		D(\bar a_i) \leq 1+(m-1)(m\eta^2+\gamma^2)+2L\gamma^2 \leq 1+m^{-2}+m^{-7}+2m^{-5} \leq 1+2m^{-2}.
	\end{equation*}
	Using $(1+x)^{-1}\geq1-x$ for $x\geq0$, we obtain
	\begin{equation}
		\label{eq:target-zero-contribution}
		\frac{1}{D(\bar a_i)} \geq1-\frac{2}{m^2}.
	\end{equation}
	If $X_{ij}=1$, then
	\begin{equation}
		\label{eq:target-one-contribution}
		D(\bar a_i) \geq L \cdot \ip{\bar a_i}{\bar q_j} = L \cdot \frac{\eta+\gamma^2}{A_iQ}  \geq L \cdot \frac{\eta}{2}=\frac{m}{2},
		\qquad\text{or}\quad
		\frac{1}{D(\bar a_i)}\leq\frac{2}{m}.
	\end{equation}
	
	Next consider any $k\neq i$.  Since the coefficient of $e_k$ is one in both
	$a_k$ and $b_i$, we have
	$$
	D(\bar a_k) \geq L \cdot \ip{\bar a_k}{\bar b_i} = L \cdot \frac{1+\gamma^2}{A_kB} \geq \frac{L}{2\sqrt m}.
	$$
	Thus all nonqueried Alice vectors together contribute at most
	\begin{equation}
		\label{eq:other-alice-contribution}
		\sum_{k\neq i}\frac{1}{D(\bar a_k)} \leq \frac{2m^{3/2}}{L} = \frac{2}{m^{3/2}}.
	\end{equation}
	
	All $L$ copies of $\bar b_i$ have the same density, denoted by $D_b$.   Clearly, $D_b\geq L$.  On the other hand,
	$$
	D_b=L+\sum_{k \in [m]}\ip{\bar b_i}{\bar a_k}+L\ip{\bar b_i}{\bar q_j}\leq L+m+L\gamma^2.
	$$
	Here, we have used the fact that every cosine similarity is at most one and
	$\ip{\bar b_i}{\bar q_j}\leq\gamma^2$.  Therefore,
	\begin{equation}
		\label{eq:b-block-contribution}
		1-m^{-2}-m^{-8}	\leq \frac{1}{1+m/L+\gamma^2} \leq L \cdot \frac{1}{D_b} \leq 1.
	\end{equation}
	
	Similarly, all $L$ copies of $\bar q_j$ have a common density $D_q$.  We again have $D_q \ge L$. Moreover,
	$$
	D_q = L+\sum_{k \in [m]}\ip{\bar q_j}{\bar a_k}+L\ip{\bar q_j}{\bar b_i} \leq L+m^{-1}+m^{-7}+L\gamma^2.
	$$
	Therefore,
	\begin{equation}
		\label{eq:q-block-contribution}
		1-m^{-4}-m^{-8}-m^{-10} \le \frac{1}{1+m^{-1}/L+m^{-7}/L+\gamma^2}
		\leq L \cdot \frac{1}{D_q}\leq1.
	\end{equation}
	
	If $X_{ij}=0$, combining~\eqref{eq:target-zero-contribution},
	\eqref{eq:b-block-contribution}, and~\eqref{eq:q-block-contribution}, while
	ignoring the nonnegative contributions of the other Alice vectors, gives
	$$
	M_{-1}(P) \geq 3-3m^{-2}-m^{-4}-2m^{-8}-m^{-10} \geq 3-\frac{4}{m^2}.
	$$
	If $X_{ij}=1$, then~\eqref{eq:target-one-contribution},
	\eqref{eq:other-alice-contribution}, \eqref{eq:b-block-contribution}, and~\eqref{eq:q-block-contribution} give
	$$
	M_{-1}(P)	\leq 2+\frac{2}{m}+\frac{2}{m^{3/2}} \leq 2+\frac{4}{m}.
	$$
	The lemma follows.
\end{proof}

\begin{theorem}
	\label{thm:quadratic-lower-bound}
	Fix any constant $\eps\in(0,1/5)$.  For any sufficiently large dimension $d$, any one-pass $1/3$-error $(1+\eps)$-approximation streaming algorithm for $M_{-1}(P)$
	for every stream $P$ of vectors in $\mathbb{R}^d$ with strictly positive pairwise cosine
	similarities requires $\Omega(d^2)$ bits of memory.  The lower bound holds even
	when the stream has length $\operatorname{poly}(d)$ and every input coordinate
	has $O(\log d)$-bit rational coordinates.
\end{theorem}

\begin{proof}
	Suppose that such a streaming algorithm uses $s$ bits.  We construct a
	one-way protocol for \Index$_{m^2}$, where $m = \lfloor(d-1)/2\rfloor$.  Alice interprets her input as a
	matrix $X\in\{0,1\}^{m\times m}$, runs the streaming algorithm on the vectors
	$a_1,\ldots,a_m$ created according to~\eqref{eq:lower-bound-alice-vector}, and sends
	the final memory configuration to Bob.  Given query index $J = (i,j)$, Bob resumes the
	algorithm and appends the $L$ copies of $b_i$ and $q_j$ as in~\eqref{eq:hard-stream}.
	
	Since $\eps<1/5$, we have $3(1-\eps)>2(1+\eps)$. Therefore, for all sufficiently large $m$, 
	$$
	(1-\eps)\left(3-\frac{4}{m^2}\right) > (1+\eps)\left(2+\frac{4}{m}\right).
	$$
	Bob chooses any threshold between the two sides of the above inequality.  By Lemma~\ref{lem:hard-instance-gap}, whenever the streaming estimate is correct, Bob determines whether $X_{ij}=0$ or $X_{ij}=1$.  The resulting protocol has error at most $1/3$ and communicates exactly the $s$-bit memory configuration.  Lemma~\ref{lem:index-lower-bound} implies $s = \Omega(m^2) = \Omega(d^2)$.
	
	The coordinate precision and stream length follow
	from~\eqref{eq:lower-bound-alice-vector}, \eqref{eq:lower-bound-bob-vectors}, and \eqref{eq:hard-stream}.
\end{proof}

\begin{remark}[An exact two-pass algorithm]
	\label{thm:two-pass-diversity}
	Under the nonnegative-cosine promise, the diversity index can be
	computed exactly in two passes using $O(d)$ words of space. In the
	first pass, the algorithm maintains the vector sum
	$
	s=\sum_{j \in [n]} u_j.
	$
	In the second pass, it computes
	$
	\sum_{i \in [n]} {1}/{\langle u_i,s\rangle} = \sum_{i \in [n]} {1}/{D_i} = M_{-1}(P).
	$
	The algorithm stores only the $d$-dimensional vector $s$ and a scalar
	counter. This contrasts with the $\Omega(d^2)$-bit one-pass lower
	bound of Theorem~\ref{thm:quadratic-lower-bound}, showing that an
	additional pass can substantially reduce the space required to compute
	the diversity index.
\end{remark}

\section{Even Density Moments under Signed Cosine}
\label{sec:signed}

When signed-cosine similarities are allowed, a density $D_i$ may be negative.
For odd integer exponents $p \ge 3$, $M_p=\sum_i D_i^p$ can have either sign and may
exhibit substantial cancellation.  We therefore focus in this section on even exponents $p=2k$, for which $M_p(P)=\sum_{i\in[n]} D_i^p = \sum_{i\in[n]} |D_i|^p$
is nonnegative and admits the usual multiplicative-approximation guarantee.
We first give a one-pass algorithm for all even moments, and then prove
a linear lower bound for $p=2$ and a $\widetilde\Omega_p(d^{p/2})$-bit lower bound for every fixed even $p\geq4$. The latter matches the polynomial dependence on $d$ of the upper bound.

\subsection{A Tensor--AMS Estimator for Even Moments}
\label{sec:signed-even-ub}

The exact degree-$p$ representation in Theorem~\ref{thm:higher-moments} uses
$\Theta_p(d^p)$ words.  For even moments, approximation lowers the
required tensor degree by a factor of two.

\begin{theorem}
	\label{thm:even-higher-moments}
	Let $p=2k$ for an integer $k\geq1$.  For every $\eps, \delta\in(0,1)$, there is a one-pass algorithm
	that returns a $(1+\eps)$-approximation to $M_p(P)$ for any input stream $P$ of vectors in $\mathbb{R}^d$ with probability at
	least $1-\delta$ using
	$$
	O\left(\binom{d+k-1}{k}\eps^{-2}\log\frac1\delta\right) = O_k\left(d^{p/2}\eps^{-2}\log\frac1\delta\right)
	$$
	words of space.
\end{theorem}

\begin{proof}
	Let	$N_{d,k}=\binom{d+k-1}{k}$,
	and define the normalized degree-$k$ monomial map
	$\psi_k:\R^d\to\R^{N_{d,k}}$ by
	$$
	(\psi_k(v))_{\vec{\alpha}} :=\sqrt{\binom{k}{{\vec{\alpha}}}}\,v^{\vec{\alpha}}, \qquad \text{where} \quad |{\vec{\alpha}}|=k.
	$$
	The multinomial theorem gives
	\begin{equation}
		\label{eq:signed-feature-identity}
		\ip{\psi_k(v)}{\psi_k(w)}=\ip{v}{w}^k.
	\end{equation}
	
	Consider one AMS sketch~\cite{AMS99}.  Choose $4$-wise independent signs
	$\sigma_1,\ldots,\sigma_n\in\{-1,+1\}$ and maintain
	$
	T:=\sum_{i\in[n]}\sigma_i\psi_k(u_i)
	$
	and
	$
	s:=\sum_{i\in[n]} u_i.
	$
	At the end of the stream, compute
	$$
	Z:=\ip{T}{\psi_k(s)}.
	$$
	By \eqref{eq:signed-feature-identity} and the identity $D_i=\ip{u_i}{s}$,
	$$
	Z=\sum_{i\in[n]}\sigma_i\ip{u_i}{s}^k =\sum_{i\in[n]}\sigma_iD_i^k.
	$$
	Therefore
	$
	\E[Z^2]=\sum_{i\in[n]}D_i^{2k}=M_p(P).
	$
	The standard AMS moment calculation gives
	$
	\Var(Z^2)\leq2M_p(P)^2.
	$
	Averaging $O(\eps^{-2})$ independent copies of $Z^2$ gives a
	$(1+\eps)$-approximation with constant success probability, and taking the
	median of $O(\log(1/\delta))$ independent averages reduces the error
	probability to $\delta$.
	
	Each AMS sketch stores one vector in $\R^{N_{d,k}}$, and all copies share the
	$d$-dimensional vector sum $s$.  The total
	space is
	$
	O\left(N_{d,k}\eps^{-2}\log\frac1\delta+d\right)=O\left(\binom{d+k-1}{k}\eps^{-2}\log\frac1\delta\right)
	$
	words.
\end{proof}

\begin{remark}
	We note that the tensor--AMS estimator fundamentally uses the parity assumption $p=2k$: it represents each contribution as
	$
	D_i^p=\langle u_i^{\otimes k},s^{\otimes k}\rangle^2,
	$
	thereby reducing the problem to estimating a second moment in the lifted tensor space. For odd $p$, an additional factor of $D_i$ remains, even when all densities are nonnegative, and the resulting expression is no longer an unweighted sum of squares. Thus, nonnegativity alone does not make the tensor--AMS estimator applicable to odd moments.
\end{remark}

\subsection{A Linear Lower Bound for $M_2$}
\label{sec:signed-second-lower}

For $p=2$, Theorem~\ref{thm:even-higher-moments} uses
$O(d\eps^{-2}\log(1/\delta))$ words.  In this section, we show that the linear dependence on $d$ is necessary.  In fact, the lower bound holds even in the nonnegative-cosine regime, and therefore also applies in the signed regime.

We note that the construction for higher even moments in the next subsection relies on the suffix contribution $m$ being lower order than the queried signal $m^{p/2}$, which requires $p>2$. The direct reduction for $p=2$ also avoids the loss of logarithmic factors.

We use the following construction of a suffix consisting of mutually orthogonal unit vectors whose vector sum is a positive multiple of the required query direction.
\begin{lemma}
	\label{hm:lem:balanced-suffix}
	Let $L$ be a power of four, let $x$ be a unit vector, and let $W$ be an
	$(L-1)$-dimensional subspace orthogonal to $x$.  There are orthonormal unit
	vectors $q_1,\ldots,q_L\in\operatorname{span}\{x\}\oplus W$ such that
	$\ip{q_\ell}{x}=1/\sqrt{L}$ for every $\ell\in[L]$, and
	$\sum_{\ell \in [L]}q_\ell=\sqrt{L} x.$
\end{lemma}

\begin{proof}
	Let $H\in\{-1,+1\}^{L\times L}$ be a Hadamard matrix whose first row
	is all $+1$.  Let $b_1=x$, and extend it by an orthonormal basis
	$b_2,\ldots,b_L$ of $W$.  For each column $\ell$ of $H$, set
	$$
	q_\ell:=\frac1{\sqrt L}\sum_{r\in [L]} H_{r\ell}b_r.
	$$
	Since $H^\top H=LI$, the vectors $q_1,\ldots,q_L$ are orthonormal.  The first
	row of $H$ gives $\ip{q_\ell}{x}=1/\sqrt L$.  Moreover, every row $r\geq 2$ of $H$ is orthogonal to the all-$+1$
	first row and hence has zero row sum. Therefore,
	$$
	\sum_{\ell\in[L]} q_\ell = \frac{1}{\sqrt{L}}\sum_{r\in[L]}
	\left(\sum_{\ell\in[L]} H_{r\ell}\right)b_r = \frac{1}{\sqrt{L}} \cdot L b_1 = \sqrt{L} x.
	$$
\end{proof}

\begin{theorem}
	\label{thm:second-density-moment}
	For every fixed $\eps\in(0,1/5)$ and every sufficiently large dimension $d$, any one-pass $1/3$-error algorithm that returns a $(1+\eps)$-approximation to $M_2(P)$ for every stream $P$ of vectors in $\mathbb{R}^d$ requires $\Omega(d)$ bits of space.  The lower bound holds even for streams of length $O(d)$ in the nonnegative-cosine regime with input vectors having $O(\log d)$-bit rational coordinates.
\end{theorem}

\begin{proof}
	We reduce from $\Index_m$, where Alice receives $z = (z_1, \ldots, z_m) \in {0,1}^m$, and Bob receives an index $j \in [m]$ and must recover $z_j$. Let $t$ be the largest power of two such that $3t^2-1\leq d,$ and set $m:=t^2$ and
	$L:=t^2.$ Thus, $2m+L-1=3t^2-1\leq d$. By the maximality of $t$, we have $m=\Theta(d)$. We construct the hard instance in $\mathbb R^{3t^2-1}$ and embed it into $\mathbb R^d$ by appending $d-(3t^2-1)$ zero coordinates to every vector. 
	
	Choose distinct vectors $e_1,\ldots,e_m, f_1,\ldots,f_m, b_2,\ldots,b_L$ from the standard basis of $\mathbb R^d$, and let
	$
	W:=\operatorname{span}\{b_2,\ldots,b_L\}.
	$
	For each $k\in[m]$, Alice inserts
	$$
	a_k:=
	\begin{cases}
		e_k,& \text{if } z_k=1,\\
		f_k,& \text{if } z_k=0.
	\end{cases}
	$$
	Her $m$ vectors are mutually orthogonal. Alice runs the streaming
	algorithm on this prefix and sends its memory configuration to Bob.
	
	Bob applies Lemma~\ref{hm:lem:balanced-suffix} with $x=e_j$,
	$W$ as defined above, and suffix length $L$. He appends orthonormal
	unit vectors $q_1,\ldots,q_L$ satisfying
	$$
	\ip{q_\ell}{e_j} = \frac{1}{\sqrt L} = \frac{1}{t}
	$$
	for every $\ell\in[L]$, where each $q_\ell$ is orthogonal to every
	$e_k$ with $k\neq j$ and to every $f_k$. All pairwise inner products
	in the full stream are therefore nonnegative.
	
	The input stream can be written as
	$$
		P(z, j) = (a_1, \ldots, a_m, q_1, \ldots, q_L)
	$$
	
	If $z_j=0$, then Alice inserted $f_j$, which is orthogonal to the
	suffix. Every stream vector has density one, and hence,
	$$
	M_2^{(0)}(P) = m+L=2t^2.
	$$
	
	If $z_j=1$, then Alice inserted $e_j$. Its density is
	$$
	D_{e_j} = 1+\sum_{\ell\in[L]}\ip{e_j}{q_\ell} = 1+\sqrt L = 1+t.
	$$
	Every other vector inserted by Alice has density one, while each
	suffix vector $q_\ell$, for $\ell\in[L]$, has density
	$$
	D_{q_\ell} = 1+\ip{q_\ell}{e_j} = 1+\frac{1}{t}.
	$$
	Consequently,
	$$
	M_2^{(1)}(P) = (m-1)+(1+t)^2 + L\left(1+\frac{1}{t}\right)^2 = 3t^2+4t+1.
	$$
	Thus,
	$$
	\frac{M_2^{(1)}(P)}{M_2^{(0)}(P)} = \frac{3t^2+4t+1}{2t^2} > \frac{3}{2}.
	$$
	
	Since $\eps<1/5$, we have $2(1+\eps)<3(1-\eps).$
	It follows that
	$
	(1+\eps)M_2^{(0)}(P) < (1-\eps)M_2^{(1)}(P).
	$
	A $(1+\eps)$-approximation to $M_2$ therefore allows Bob to recover
	$z_j$. Hence, the memory configuration sent by Alice must contain
	$\Omega(m)=\Omega(d)$ bits.
	
	Finally, the stream length is
	$
	m+L=2t^2=\Theta(d).
	$
	The prefix vectors are standard basis vectors. Since $L=t^2$ is a
	power of four, let $H\in\{-1,+1\}^{L\times L}$ be the Hadamard matrix
	used in Lemma~\ref{hm:lem:balanced-suffix}, with its first row being all-$+1$. The construction in that lemma gives
	$
	q_\ell = \frac{1}{t} \left(e_j+\sum_{r=2}^L H_{r\ell}b_r\right).
	$
	Therefore, every coordinate of every suffix vector belongs to
	$\{0,\pm1/t\}$. Since $t=\Theta(\sqrt d)$, all input coordinates are
	rational and can be represented using $O(\log d)$ bits.
\end{proof}

\subsection{Lower Bounds for Higher Even Moments}
\label{subsec:signed-higher-lower}

We now prove the lower bound for every fixed even $p\geq4$.  
The proof uses antipodal pairs to cancel Alice's prefix sum
and a low-correlation sign code to encode many possible query directions.  We use an elementary low-correlation-code argument; related code-based membership reductions appear in lower bounds for subspace sketches, for example in \cite{LWW20}. Our main new ingredient is a construction that realizes the query direction as the sum of a mutually orthogonal suffix of vectors.

\begin{lemma}
	\label{lem:random-code}
	Fix a constant $\zeta>0$.  For every sufficiently large integer $m$, there is a
	set $\mathcal S\subseteq\{-1,+1\}^m$ with $|\mathcal S|\geq m^\zeta$ such that,
	for all distinct $a,b\in\mathcal S$,
	\begin{equation}
		\label{eq:code-correlation}
		|\ip{a}{b}|\leq C_\zeta\sqrt{m\log m},
	\end{equation}
	where $C_\zeta>0$ depends only on $\zeta$.
\end{lemma}

\begin{proof}
	Choose $N= \lceil m^\zeta \rceil$ independent uniformly random vectors from
	$\{-1,+1\}^m$.  For a fixed pair $a,b$, the products $a_rb_r\ (r \in [m])$ are independent
	Rademacher random variables, so Hoeffding's inequality gives
	$$	
	\Pr\left[|\ip{a}{b}|>C_\zeta \sqrt{m\log m}\right]\leq2m^{-C^2/2}.
	$$
	There are fewer than $N^2$ pairs.  To apply the union bound, we set
	$C_\zeta:=2\sqrt{\zeta+1}$ to make the total error probability less than one.  Hence, a realization satisfying
	\eqref{eq:code-correlation} exists. 
\end{proof}


\begin{theorem}
	\label{thm:signed-lower-bound}
	Fix an even integer $p\geq4$ and a constant $\eps\in(0,1)$.  For every sufficiently large dimension $d$, any
	one-pass $1/3$-error streaming algorithm that returns a $(1+\eps)$-approximation to
	$M_p(P)$ for every stream $P$ of vectors in $\mathbb{R}^d$ under signed-cosine similarity
	requires
	$
	\Omega_{p,\eps}\left(\frac{d^{p/2}}{(\log d)^{p/2}}\right) = \widetilde \Omega_{p,\eps}\left(d^{p/2}\right)
	$
	bits of space.  The lower bound holds even for
	streams of length $\widetilde O_p(d^{p/2})$ with input vectors having $O(\log d)$-bit rational coordinates.
\end{theorem}

\begin{proof}
	Fix an arbitrary sufficiently large dimension $d$, and let $m$
	be the largest power of four such that $2m-1\leq d.$ By the maximality of $m$, we have $m=\Theta(d)$.	We construct the hard instance in $\mathbb R^{2m-1}$ and embed it into $\mathbb R^d$ by appending $d-(2m-1)$ zero coordinates to
	every vector. 
	
	Let $\zeta:=p+3$ and let $\mathcal S\subseteq\{-1,+1\}^m$ be the family from
	Lemma~\ref{lem:random-code}. For $a\in\mathcal S$, write
	$
	\bar a:=\frac{a}{\sqrt m}.
	$
	Every $\bar a$ is a unit vector, and for any pair of distinct vectors $a,b\in\mathcal S$,
	\begin{equation}
		\label{eq:normalized-code-correlation}
		\left|\ip{\bar a}{\bar b}\right| \leq C_\zeta\sqrt{\frac{\log m}{m}},
	\end{equation}
	for a sufficiently large constant $C_\zeta=C_\zeta(p)$.
	
	Set $\gamma_\eps:=\frac{1-\eps}{1+\eps}.$ Choose a sufficiently small constant $c=c(p,\eps)>0$ such that
	$
	c \cdot C_\zeta^p\leq\frac{\gamma_\eps}{8},
	$
	and let
	$
	R:=\Big\lfloor c\left(\frac{m}{\log m}\right)^{p/2}\Big\rfloor.
	$
	For every $R$-element subset $T\subseteq\mathcal S$, Alice inserts the antipodal pair $\{\bar a,-\bar a\}$ for every $a\in T$.  Note that the sum of her vectors is exactly
	zero.
	
	Now fix a query $x\in\mathcal S$. Bob uses a fixed orthogonal auxiliary $(m-1)$-dimensional subspace
	and applies Lemma~\ref{hm:lem:balanced-suffix} with $L=m$ and the unit vector
	$\bar x=x/\sqrt m$.  He appends the resulting orthonormal vectors
	$q_1,\ldots,q_m$, whose sum is
	$
	s:=\sum_{\ell \in [m]} q_\ell=\sqrt m\,\bar x.
	$
	
	Let the input stream be
	$$
	P_{T, x} := ((\bar a_1,-\bar a_1), \ldots, (\bar a_{|T|},-\bar a_{|T|}), q_1, \ldots, q_m).
	$$
	The vectors associated with $a\in T$ therefore have densities
	$D_{\bar a}=\sqrt m\,\ip{\bar a}{\bar x}$ and $D_{-\bar a}=-\sqrt m\,\ip{\bar a}{\bar x}$, while every suffix vector $q_\ell\ (\ell \in [m])$ has density one. 
	
	Since $p$ is even, the final moment is exactly
	\begin{equation}
		\label{hm:eq:signed-moment-identity}
		M_p(P_{T,x})=2m^{p/2} \sum_{a\in T}|\ip{\bar a}{\bar x}|^p + m.
	\end{equation}
	If $x\in T$, the self-correlation term gives
	\begin{equation}
		\label{hm:eq:yes-H}
		 \sum_{a\in T}|\ip{\bar a}{\bar x}|^p  \geq1.
	\end{equation}
	If $x\notin T$, then \eqref{eq:normalized-code-correlation} and
	the choice of $R$ imply
	\begin{equation}
		\label{hm:eq:no-H}
		 \sum_{a\in T}|\ip{\bar a}{\bar x}|^p  \leq R \cdot C_\zeta^p\left(\frac{\log m}{m}\right)^{p/2} \leq c \cdot C_\zeta^p \leq \frac{\gamma_\eps}{8}.
	\end{equation}
	Since $p>2$, for a sufficiently large $m$, $m^{1-p/2}\leq\frac{\gamma_\eps}{4}.$
	Combining this inequality with
	\eqref{hm:eq:signed-moment-identity}--\eqref{hm:eq:no-H} gives
	$$
		M_p(P_{T,x}) \le \frac{\gamma_\eps}{2}\,m^{p/2} \ \ \text{if} \ \  x\notin T, \quad \text{and} \quad
		M_p(P_{T,x}) \ge 2m^{p/2} \ \  \text{if} \ \ x\in T. 
	$$
	In the first case, the $(1+\eps)$-approximation of $M_p$ is at most
	$
	(1+\eps)\frac{\gamma_\eps}{2}m^{p/2}=\frac{1-\eps}{2}m^{p/2},
	$
	whereas in the second case, it is at least
	$
	2(1-\eps)m^{p/2}.
	$
	Thus, the estimate determines whether $x\in T$.
	
	By taking a subset if necessary, we may assume that
	$
	N:=|\mathcal S| = \left\lceil m^\zeta\right\rceil.
	$
	Let
	$$
	\mathcal F:=\{T\subseteq\mathcal S:|T|=R\}.
	$$
	Suppose the streaming algorithm uses $b$ bits of space. Let
	$
	k_0:=\left\lceil C_{k_0}\log N\right\rceil,
	$
	where $C_{k_0}$ is a sufficiently large constant. Alice runs $k_0$ independent copies of the algorithm on her prefix and sends their final memory configurations to Bob, using a total of $k_0 b$ bits. For each $x\in\mathcal S$, Bob makes a
	fresh copy of every received configuration, appends only the suffix
	corresponding to $x$, and takes the median of the resulting estimates.
	Thus, every membership query is evaluated from the same prefix states,
	and evaluating one query does not affect any other. For each fixed
	$x\in\mathcal S$, a Chernoff bound implies that the probability of
	incorrectly determining whether $x\in T$ is at most $N^{-3}$. A union
	bound over all $N$ queries then shows that, with probability at least
	$1-N^{-2}\geq 0.99$, every membership query is answered correctly.
	Consequently, Bob reconstructs $T$ exactly.
	
	Fix random coins for which at least a $0.99$ fraction of the sets in $\mathcal F$ are decoded correctly.  A deterministic message of $k_0 b$ bits has at most $2^{k_0 b}$ values, and one message cannot correctly decode two
	distinct sets.  Therefore $2^{k_0 b}\geq0.99|\mathcal F|$, or 
	$$
	b \ge \frac{1}{k_0}(\log |\mathcal F| - O(1)) .
	$$
	Since $N/R$ is a positive polynomial in $m$,
	$$
	\log|\mathcal F| =\log\binom NR \geq R\log\frac NR =\Omega_p(R\log m).
	$$
	As $k_0=O_p(\log m)$, it holds that
	$$
	b =\Omega_{p}(R) =\Omega_{p,\eps}\left( \frac{m^{p/2}}{(\log m)^{p/2}}\right).
	$$
	Since $m=\Theta(d)$, the preceding lower bound becomes
	$
	b = \Omega_{p,\eps}\left( \frac{d^{p/2}}{(\log d)^{p/2}} \right) = \widetilde{\Omega}_{p,\eps}(d^{p/2}).
	$
	The prefix contains $2R$ vectors and the suffix contains $m$
	vectors, so the stream length is
	$
	2R+m = \widetilde{O}_{p,\eps}(d^{p/2}).
	$
	Since $m$ is a power of four, $\sqrt m$ is an integer, and the
	vectors $\bar a$ have coordinates $\pm1/\sqrt m$ with
	$O(\log d)$-bit rational representations. The Hadamard construction
	in Lemma~\ref{hm:lem:balanced-suffix} also uses
	$O(\log d)$-bit rational coordinates.
\end{proof}

\section{Higher Density Moments under Nonnegative Cosine: Upper Bound}
\label{sec:nonnegative}

We now consider higher density moments in the nonnegative-cosine regime, where every pair of vectors in $P=(u_1,\ldots,u_n)$ has nonnegative cosine similarity. The main new case is that of odd integers $p>2$, since even moments are already covered by the sign-robust tensor--AMS estimator of the preceding section. Nevertheless, the algorithm developed here applies more generally to every fixed real exponent $p>2$. Nonnegativity ensures that all densities are positive and gives a
moment-ratio bound strong enough to support uniform subsampling. The corresponding lower bound, which is the main technical contribution of the paper, is presented in
Section~\ref{sec:nonnegative-lb}.

The algorithm has two ingredients.  First, a structural moment-ratio bound shows that a uniform sample of roughly $d^{p/2}$ density coordinates is sufficient. Second, a linear $F_p$-sketch of the sampled density vector can be maintained
without storing the sampled vectors themselves.

Throughout the section, let $s:=\sum_{i \in [n]} u_i$ and $D_i=\ip{u_i}{s}$.
In the nonnegative-cosine regime, we have $D_i\geq 1$ for every $i \in [n]$.

We use the following standard result of frequency-moment sketching.
For every fixed $p>2$, every dimension $K$, and every $\eta\in(0,1)$, there is a linear sketch for a vector $y\in\mathbb R^K$ that returns a $(1+\eta)$-approximation to
$
F_p(y):=\sum_{r\in[K]} |y_r|^p
$
with success probability $0.99$ using
\begin{equation}
	\label{eq:black-box-fp-space}
	O_p\left(K^{1-2/p}\eta^{-2} + K^{1-2/p}\eta^{-4/p}\log K \right) = \widetilde O_p\left(K^{1-2/p}\eta^{-2}\right)
\end{equation}
counters~\cite{Ganguly15}, where $\widetilde O_p(\cdot)$ hides logarithmic factors and constants depending on $p$. The sketch can be represented by a random linear map
$\Pi\in\mathbb R^{m\times K}$ and a decoder $\mathsf{Dec}_p$ such that
$\mathsf{Dec}_p(\Pi y)$ approximates $F_p(y)$.  

The following moment-ratio bound enables uniform sampling.

\begin{lemma}
	\label{lem:moment-ratio}
	In the nonnegative-cosine regime, for every real $p\geq1$,
	\begin{equation*}
		\label{eq:moment-ratio}
		\frac{nM_{2p}(P)}{M_p(P)^2} \leq \rank(G)^{p/2} \leq d^{p/2},
	\end{equation*}
	where $G$ is the cosine Gram matrix of the stream.
\end{lemma}

\begin{proof}
	Since $s=\sum_i u_i$ and $\norm{u_i}_2=1$,
	$
	M_1=\sum_{i \in [n]}\ip{u_i}{s}=\norm{s}_2^2
	$
	and
	$
	D_i=\ip{u_i}{s}\leq\norm{s}_2=\sqrt{M_1}.
	$
	Consequently,
	\begin{equation}
		\label{eq:m2p-via-mp}
		M_{2p} = \sum_i D_i^{2p} \leq \left(\max_i D_i\right)^p\sum_iD_i^p \leq M_1^{p/2}M_p.
	\end{equation}
	By H\"older’s inequality, we have
	\begin{equation}
		\label{eq:mp-via-m1}
		M_p \geq n^{1-p}M_1^p.
	\end{equation}
	Moreover, Cauchy--Schwarz and Lemma~\ref{lem:rank-bound} give
	$$
	n^2 =\left(\sum_{i \in [n]} 1\right)^2 \leq \left(\sum_{i \in [n]}D_i\right) \left(\sum_{i \in [n]}\frac1{D_i}\right)=M_1 \cdot M_{-1}(P) \leq \rank(G) \cdot M_1,
	$$
	which implies
	\begin{equation}
		\label{eq:m1-lower-bound}
		M_1\geq\frac{n^2}{\rank(G)}.
	\end{equation}
	Combining \eqref{eq:m2p-via-mp}, \eqref{eq:mp-via-m1}, and
	\eqref{eq:m1-lower-bound}, and using Lemma~\ref{lem:rank-bound}, we obtain
	$$
	\frac{nM_{2p}}{M_p^2} \leq \frac{nM_1^{p/2}}{M_p} \leq \frac{n^p}{M_1^{p/2}} \leq \rank(G)^{p/2} \leq d^{p/2}.
	$$
\end{proof}

\vspace{2mm}
\noindent{\bf Algorithm and Analysis.\ }
We first describe an estimator with constant success probability, assuming that the stream length $n$ is known.  The algorithm is described in Algorithm~\ref{alg:approx-higher-moment}.

\begin{algorithm}[t]
	\caption{One-pass estimator for $M_p\ (p > 2)$ in the nonnegative-cosine regime}
	\label{alg:approx-higher-moment}
	\DontPrintSemicolon
	\KwIn{A stream $P = (x_1, \ldots, x_n)$ of points in $\mathbb{R}^d$ in the nonnegative-cosine regime, parameters $p>2$ and $\eps \in(0,1)$, the stream length $n$.}
	\KwOut{A $(1+\eps)$-approximation to $M_p(P)$}
	
	Set $\theta \gets  C_p{d^{p/2}}/{\eps^2}$, $q \gets \min\left\{1,\frac{\theta}{n}\right\}$, and
	$K\gets \lceil 4\theta \rceil$\;
	
	Initialize an independent linear $F_p$-sketch $(\Pi,\mathsf{Dec}_p)$ on
	$K$ coordinates with accuracy $\eps/4$, and let $m$ be the number of counters in the $F_p$-sketch (or, the number of rows of $\Pi$)\; 
	
	Initialize
	$s\gets \mathbf{0}\in\mathbb R^d$, $Y\gets \mathbf{0}\in\mathbb R^{m\times d}$, and $r\gets0$\;
	
	\ForEach{incoming $x_t\ (t = 1, \ldots, n)$}{
		$u_t\gets x_t/\norm{x_t}_2$\;
		$s\gets s+u_t$\;
		Draw $b_t\sim\operatorname{Bernoulli}(q)$\; \label{ln:b}
		\If{$b_t=1$}{
			$r\gets r+1$\;
			\lIf{$r>K$}{\KwRet{\textnormal{\textsc{Overflow}}}}
			Generate the $r$-th sketch column $a_r=\Pi e_r$ and update
			$Y\gets Y+a_ru_t^\top$\;
		}
	}
	Compute $z\gets Ys$ and $\widetilde F\gets\mathsf{Dec}_p(z)$\;
	\Return{$\widetilde M_p\gets \widetilde F/q$}\;
\end{algorithm}

Let $\theta :=  C_p{d^{p/2}}/{\eps^2}$, and $q:=\min\left\{1,\frac{\theta}{n}\right\}$,
where $C_p>0$ is a sufficiently large constant.  Independently sample each stream index with probability $q$.  If the sampled indices are $i_1,\ldots,i_R$ for some $R \le K := \lceil 4\theta \rceil$, define the sampled density vector as
$$
y=(D_{i_1},\ldots,D_{i_R},0,\ldots,0)^T \in \mathbb R^K.
$$
If more than $K$ points are sampled, then we report overflow.
Clearly,
\begin{equation}
	\label{eq:nonnegative-Fp}
	F_p(y)=\sum_{r\in [R]} D_{i_r}^p = \sum_{i \in [n]} b_iD_i^p,
\end{equation}
where $b_i = 1$ if index $i$ is sampled, and $b_i = 0$ otherwise.

The algorithm does not store $u_{i_1},\ldots,u_{i_R}$, since otherwise we need to use $\Omega(d^{p/2+1})$ space.  Instead, let $\Pi\in\mathbb R^{m\times K}$ be the matrix of a linear $F_p$-sketch and let $a_r$ denote its $r$th column.  The algorithm maintains
$
Y:=\sum_{r \in [R]} a_r u_{i_r}^{\top}\in\mathbb R^{m\times d}
$
and
$
s=\sum_{i \in [n]} u_i\in\mathbb R^d.
$
At the end of the stream,
\begin{equation}
	\label{eq:lifted-sketch-identity}
	Ys = \sum_{r \in [R]} a_r\ip{u_{i_r}}{s} = \sum_{r \in [R]} a_rD_{i_r} = \Pi y.
\end{equation}
In this sense, each scalar counter of the original linear sketch is lifted to a $d$-dimensional vector, allowing the desired sketch $\Pi y$ to be recovered once the final sum $s$ is known.

\begin{lemma}
	\label{lem:uniform-moment-sampling}
	Let $b_1,\ldots,b_n$ be independent Bernoulli variables with parameter $q$ as
	in Line \ref{ln:b} of Algorithm~\ref{alg:approx-higher-moment}, and define
	$
	X:=\frac{1}{q}\sum_{i \in [n]} b_iD_i^p.
	$
	Then $\E[X]=M_p(P)$ and $\Var[X] \leq \frac{\eps^2}{C_p} \cdot M_p(P)^2.$
\end{lemma}

\begin{proof}
	The fact that $\E[X]=M_p(P)$ is immediate.  If $q=1$, the estimator is exact.  Otherwise, $q=\theta/n$, and independence gives
	$$
	\Var[X] = \frac{1-q}{q}\sum_{i \in [n]}D_i^{2p} \leq \frac{M_{2p}(P)}q = \frac{n}{\theta} M_{2p}(P).
	$$
	Therefore, by Lemma~\ref{lem:moment-ratio},
	$$
	\frac{\Var[X]}{M_p(P)^2} \leq \frac{n}{\theta}\frac{M_{2p}(P)}{M_p(P)^2} \leq \frac{d^{p/2}}{\theta} \leq \frac{\eps^2}{C_p},
	$$
	where the last inequality follows from the
	definition of $\theta$.
\end{proof}

\begin{theorem}
	\label{thm:approx-higher-moments}
	For every fixed real $p>2$, in the nonnegative-cosine regime, there is a
	one-pass algorithm that, for every $\eps, \delta\in(0,1)$, returns a $(1+\eps)$-approximation to $M_p(P)$ for any input stream $P$ of vectors in $\mathbb{R}^d$ with probability at least $1-\delta$.
	Its space usage is
	$
	\widetilde O_p\left( d^{p/2}\eps^{-4+4/p}\log\frac1\delta \right)
	$
	words. 
\end{theorem}

\begin{proof}
	Consider one instance of Algorithm~\ref{alg:approx-higher-moment}.  By
	Lemma~\ref{lem:uniform-moment-sampling} and Chebyshev's inequality, choosing $C_p$
	sufficiently large makes
	\begin{equation}
		\label{eq:nonnegative-Mp}
		\frac{1}{q}\sum_{i \in [n]} b_iD_i^p=(1\pm\eps/4)M_p(P)
	\end{equation}
	with probability at least $9/10$.  The expected number of sampled points is at
	most $\theta$.  A Chernoff bound shows that the overflow event
	$\sum_{i\in[n]} b_i>4\theta$ happens with probability at most $e^{-\Omega(\theta)}$ when $q<1$; when $q=1$, overflow never occurs because $n\leq\theta$.
	
	Condition on the sampled set and on no overflow.  By
	\eqref{eq:lifted-sketch-identity}, the vector supplied to the $F_p$ decoder is
	exactly $\Pi y$, where $y$ is the sampled density vector.  The black-box $F_p$-sketch
	therefore returns
	\begin{equation}
		\label{eq:nonnegative-tilde-F}
		\widetilde F=(1\pm\eps/4)F_p(y)
	\end{equation}
	with constant probability.  
	Combining  \eqref{eq:nonnegative-tilde-F}, \eqref{eq:nonnegative-Fp}, and \eqref{eq:nonnegative-Mp}, 
	we have
	$
	{\widetilde F}/{q}=(1\pm\eps)M_p(P).
	$
	Running
	$O(\log(1/\delta))$ independent instances of Algorithm~\ref{alg:approx-higher-moment} and returning their median reduces the
	error probability to $\delta$.
	
	To bound the space usage, we apply the $F_p$-sketch to a
	$K$-dimensional vector, where
	$
	K=\Theta(d^{p/2}\eps^{-2}),
	$
	with relative error $\Theta(\eps)$. By
	\eqref{eq:black-box-fp-space}, the sketch uses
	$
	\widetilde O_p\left(K^{1-2/p}\eps^{-2}\right)
	$
	scalar counters. In the lifted summary, each scalar counter is replaced
	by a $d$-dimensional vector, and we additionally store the vector sum
	$s$. Hence, one instance of
	Algorithm~\ref{alg:approx-higher-moment} uses
	$
	\widetilde O_p\left( dK^{1-2/p}\eps^{-2}+d \right) = \widetilde O_p\left( d^{p/2}\eps^{-4+4/p} \right)
	$
	words. Finally, independent repetition and taking the median amplify the
	success probability to $1-\delta$, increasing the space by a factor of
	$O(\log(1/\delta))$.
\end{proof}

\begin{remark}[Unknown stream length]
	We note that the assumption that $n$ is known is not essential.  Given an upper bound $N$ on the stream length, maintain parallel copies at geometric sampling rates
	$q_\ell=2^{-\ell}$ for $\ell=0,1,\ldots,\log N$.  Mark a level as overflowed if it stores more than $4\theta$ points.  After observing $n$, choose the level $\ell^*$ determined solely by $n$ such that $\theta/2\leq q_{\ell^*}n\leq\theta$; when $n\leq\theta$, use sampling rate $q=1$. If the level $\ell^*$ overflowed, the copy reports failure.  Its overflow probability is $e^{-\Omega(\theta)}$, and the variance proof changes only by a constant factor.  The parallel levels contribute one additional logarithmic factor, which will again be subsumed by the $\widetilde O_p(\cdot)$ notation.
\end{remark}

\section{Higher Density Moments under Nonnegative Cosine: Lower Bound}
\label{sec:nonnegative-lb}

We now prove a superlinear lower bound for higher density moments under nonnegative cosine.

\begin{theorem}
	\label{thm:main}
	For every fixed integer $p>2$, there is a constant $\eps_p>0$ such that, for all sufficiently large dimensions $d$, any one-pass $1/3$-error streaming algorithm that gives a $(1+\eps_p)$-approximation to $M_p(P)$ for every input stream $P$ of unit vectors in $\mathbb{R}^d$ with pairwise nonnegative cosine similarities uses at least
	$\Omega_p(d^{\rho_p})$ bits of space, where
	$
	\rho_p=\frac{p^2(p-1)}{2(p^2-2)} \in (\frac{p-1}{2}, \frac{p}{2}).
	$
	The lower bound holds even when the input stream has length polynomial in $d$ and its vector coordinates can be represented using $O_p(\log d)$ bits.
\end{theorem}

The remainder of this section proves Theorem~\ref{thm:main}. We begin
by introducing the notation and parameters used throughout the proof.
We then develop the two main technical ingredients: a sparse
nonnegative code and a finite-difference identity. Building on these
tools, we construct the hard input instances and establish their key
properties. Finally, we reduce the one-way \Index\ problem to
approximating $M_p$ and analyze the finite-precision requirements of the
reduction.

\subsection{Notation and Parameters}

We collect the main notation used throughout the proof.

\begin{itemize}
	\item $p>2$ is a fixed integer exponent, and
	$\rho_p:=\frac{p^2(p-1)}{2(p^2-2)}$. The base dimension parameter is $m$, and
	$N=m \lfloor c_Nm^{\rho_p-1} \rfloor$ is the number of bits encoded by Alice, where
	$c_N=c_N(p)>0$ is a sufficiently small constant. The construction
	has dimension $4m$ before a one-dimensional perturbation for meeting the precision requirement.
	
	\item $a_1,\ldots,a_N\in\R_{\geq0}^m$ are unit codewords. A balanced
	coloring $\chi:[N]\to[m]$ assigns each codeword to a color, with
	$|\chi^{-1}(c)|=N/m$ for every color $c\in[m]$.
	
	\item $e_1,\ldots,e_m$, $g_1,\ldots,g_m$, and
	$w_1,\ldots,w_m$ are orthonormal bases for the anchor, balancing, and
	suffix blocks, respectively. These three blocks are mutually
	orthogonal and are also orthogonal to the code block.
	
	\item For $j=0,\ldots,p$, $c_j:=(-1)^{p-j}\binom pj$ and $t_j:=t_0+jh$.
	Here $c_j$ is the $j$-th finite-difference coefficient, while $t_0$
	and $h$ are the initial coefficient and level spacing.
	
	\item $\lambda$ is the local-anchor coefficient, and
	$\mu_j:=\sqrt{1-\lambda^2-t_j^2}$
	is the balancing coefficient that makes each gadget vector a unit
	vector.
	
	\item $\gamma_p,\beta_p>0$ are the coefficients in the common gadget
	sum 
	$
	\sum_{u\in\Gzero(a_i)}u = \sum_{u\in\Gone(a_i)}u = \gamma_pa_i+\beta_pe_{\chi(i)}
	$.
	The constant $\kappa_p:=2p!\,h^p$ is the coefficient of the isolated $p$-th-order response.
	
	\item $B_i$ is the \emph{local baseline}, i.e., the local-anchor contribution to the density of vectors
	in the $i$-th gadget, and $Z_i$ is its code-space response. Thus, a
	level-$j$ vector in that gadget has density
	$B_i+t_jZ_i$.
	
	\item $L$ is the number of copies of each suffix vector, and
	$Q:=L\sqrt m$ is the amplitude added by Bob in the queried code direction. We take
	$L=\Theta(m^{(p-2)/2})$, and hence $Q=\Theta(m^{(p-1)/2})$.
	
	\item For each codeword, let $R_i:=1+\sum_{r\neq i}\ip{a_i}{a_r}$ and
	$R_{\max}:=\left(N/m\right)^{(p-1)/(p-2)}$.
	The sparse code guarantees $R_i=O_p(R_{\max})$.
	
	\item $V_0(q)$ denotes the value of $M_p$ on the all-zero encoding
	followed by Bob's suffix for query $q$. It is the global baseline in	the final reduction.
\end{itemize}

\subsection{A Sparse Nonnegative Code}

We need many nonnegative unit vectors whose $p$-th power
cross-correlations are small, while their row sums remain
controlled.

\begin{lemma}
	\label{lem:code}
	Fix an integer $p>2$ and a constant $\iota \in (0, \frac{1}{2})$.  There is a constant
	$C=C(p,\iota)$ such that the following holds.  Suppose $m$ is
	sufficiently large and $m<N=o(m^{p/2}).$
	Then there exist $N$ distinct unit vectors $a_1,\ldots,a_N\in\R^m_{\geq 0}$ satisfying, for every $i\in[N]$,
	$$
	\sum_{j\neq i}\ip{a_i}{a_j}^p \leq \iota, \quad \text{and} \quad \sum_{j\neq i}\ip{a_i}{a_j} \leq C\left(\frac Nm\right)^{(p-1)/(p-2)}.
	$$
	Moreover, each $a_i$ may be chosen as a normalized incidence vector of
	an $s$-subset of $[m]$, where $s=\Theta_{p,\iota}\left(\left(N/m\right)^{1/(p-2)}\right)$ and $s=o(\sqrt m)$.
\end{lemma}

\begin{proof}
	Let $M := 4N$ and 
	$
	s:=\left\lceil C_0\left({4N}/{m}\right)^{1/(p-2)} \right\rceil
	$,  
	where $C_0 = C_0(p, \iota)$ is a sufficiently large constant. Since $N=o(m^{p/2})$, we have $s=o(\sqrt m)$.
	
	Independently choose $M$ uniformly random $s$-subsets $S_1,\ldots,S_M\subseteq[m]$ and set $a_i={\ind_{S_i}}/{\sqrt s}.$
	For two independent random supports $S$ and $T$, let
	$X=|S\cap T|$.  Then $\ip{a_S}{a_T}= X/s$, and $\E[X]={s^2}/{m}$.
	By the factorial-moment identity for a hypergeometric random variable (Lemma~\ref{lem:hypergeometric-moment} in Appendix~\ref{sec:lem:sparse-intersection}), we have
	$
	\E[X^p]\leq C_p \cdot {s^2}/{m}
	$
	for a constant $C_p$.  Hence
	\begin{equation*}\label{eq:pair-p-expectation}
		\E\left[\ip{a_S}{a_T}^p\right] \leq \frac{C_p}{m s^{p-2}},
		\quad \text{and} \quad
		\E\left[\ip{a_S}{a_T}\right]=\frac sm,
	\end{equation*}
	which implies
	\begin{align*}
		\mathbb E\left[\frac1M\sum_i \sum_{j\neq i}\ip{a_i}{a_j}\right]
		&\leq \frac{Ms}{m} = O_{p,\iota}\left(\left(\frac Nm\right)^{(p-1)/(p-2)}
		\right), \quad \text{and}\\
		\mathbb E\left[\frac1M\sum_i \sum_{j\neq i}\ip{a_i}{a_j}^p\right]
		&\leq \frac{C_pM}{m s^{p-2}} = O_p\left(C_0^{-(p-2)}\right).
	\end{align*}
	By Markov's inequality and a union bound, with positive probability
	both total sums are at most four times their respective expectations.
	Fix such a realization.  Delete every row whose
	first power sum exceeds $16$ times the corresponding expected
	average, and also delete every row whose $p$-th power sum exceeds $16$ 
	times its expected average.  At most $M/4$ rows are deleted for each
	condition, so at least $M/2\geq N$ rows remain.  Take any $N$ of them.
	Deleting vectors only decreases the row sums of the remaining vectors.
	Finally, choose $C_0$ sufficiently large that the remaining
	$p$-th power row sums are at most $\iota$.  The first power bound follows
	with a sufficiently large constant $C=C(p,\iota)$.
	
	Finally, note that if two unit vectors were identical, their inner product would be $1$, forcing the relevant $p$-power row sum to be at least $1$, contradicting the chosen bound $\iota<1/2$.
\end{proof}

\subsection{The Finite-Difference Construction}

For $j=0,1,\ldots,p$, define $c_j:=(-1)^{p-j}\binom pj.$ Let $\cP:=\{j:c_j>0\}$ and
$\cN:=\{j:c_j<0\}$ be the sets of positive and negative finite-difference coefficients, respectively.

We use the following constants: $\lambda^2=\frac23$, $t_0=\frac1{20}$, and $h=\frac{2}{5p}.$
For each $j = 0, 1, \ldots, p$, let $t_j=t_0+jh$, and $\mu_j=\sqrt{1-\lambda^2-t_j^2}$.
Then $t_p=9/20$, and all $\mu_j$'s are positive because $\lambda^2+t_p^2<1$.

The quantity $t_j$ is the coefficient placed on a codeword at level $j$.  The parameter $t_0$ is its smallest value and $h$ is the level spacing.  Later, if a final stream sum produces code response $Z$ and a common {\em local baseline} $B$ (contributed by the local anchors), then a level-$j$ vector has density $B+t_jZ$.  The finite-difference coefficients $c_j$ are chosen precisely so that all powers below $p$ cancel (see, e.g., \cite[Chapter~2.6]{GKP94}).

More precisely, we have the following lemma, whose proof is included in Appendix~\ref{sec:proof:lem:finite-difference} for completeness.

\begin{lemma}[Finite-difference identity]
	\label{lem:finite-difference}
	For every $B,Z\in\R$ and every integer $0\leq r\leq p$,
	$$
	\sum_{j=0}^p c_j(B+t_jZ)^r =
	\begin{cases}
		0,&0\leq r<p,\\
		p!\,h^pZ^p,&r=p.
	\end{cases}
	$$
\end{lemma}

\subsection{Local Anchors and Nonnegative Geometry}

We work in the orthogonal direct sum of a code block, an anchor block, and
a balancing block.  For a codeword $a_i$, a level $j$, a balanced color function $\chi:[N]\to[m]$, and a sign $\sigma\in\{-1,+1\}$, define
\begin{equation}
	\label{eq:gadget-vector}
	u_{i,j,\sigma} = t_j a_i+\lambda e_{\chi(i)}+\sigma\mu_j g_{\chi(i)}.
\end{equation}
We call $u_{i,j,\sigma}$ a level-$j$ vector associated with the codeword $a_i$; the sign $\sigma\in\{-1,+1\}$ specifies its balancing component. By construction, $\norm{u_{i,j,\sigma}}=1$.

\begin{lemma}
	\label{lem:acute}
	Every pair of vectors of the form \eqref{eq:gadget-vector} has nonnegative inner product.  
\end{lemma}

\begin{proof}
	If $\chi(i)\neq\chi(i')$, then the local anchor and balancing blocks are
	orthogonal, so
	$$
	\ip{u_{i,j,\sigma}}{u_{i',k,\tau}}=t_jt_k\ip{a_i}{a_{i'}}\geq0
	$$
	because the codewords have nonnegative coordinates.
	
	If $\chi(i)=\chi(i')$, then
	$$
	\ip{u_{i,j,\sigma}}{u_{i',k,\tau}}=t_jt_k\ip{a_i}{a_{i'}}+\lambda^2+\sigma\tau\mu_j\mu_k.
	$$
	Since $\sigma\tau \ge -1$, we have
	$$
	\ip{u_{i,j,\sigma}}{u_{i',k,\tau}} \geq \lambda^2-\max_{\ell}\mu_\ell^2=\frac{2}{3}-\left(\frac{1}{3}-t_0^2\right)=\frac{1}{3}+t_0^2>0.
	$$
\end{proof}

We now define two gadgets (that is, two multisets of vectors) for each codeword $a_i$.  The one-gadget is defined as
\begin{equation}
	\label{eq:G1}
	\Gone(a_i) = \biguplus_{j\in\cP} \left\{|c_j|\text{ copies of }u_{i,j,+},\ |c_j|\text{ copies of }u_{i,j,-}
	\right\},
\end{equation}
and the zero-gadget is
\begin{equation}
	\label{eq:G0}
	\Gzero(a_i) = \biguplus_{j\in\cN} \left\{|c_j|\text{ copies of }u_{i,j,+},\ |c_j|\text{ copies of }u_{i,j,-}
	\right\}.
\end{equation}
Both gadgets contain exactly $2^p$ stream vectors, counting multiplicity, since $\sum_{j\in \cP} c_j=\sum_{j\in \cN}|c_j|=2^{p-1}$ and each level appears with both balancing signs.

The following two lemmas are the key properties of our hard input constructions.

\begin{lemma}\label{lem:equal-sums}
	For every $i \in [N]$,
	$
	\sum_{u\in\Gone(a_i)}u = \sum_{u\in\Gzero(a_i)}u.
	$
	Moreover, there are positive constants $\gamma_p,\beta_p$ (depending
	only on $p$) such that this common sum is
	$
	\gamma_p a_i+\beta_p e_{\chi(i)}.
	$
\end{lemma}

\begin{proof}
	The two signs cancel the balancing component:
	$$
	u_{i,j,+}+u_{i,j,-} = 2t_j a_i+2\lambda e_{\chi(i)}.
	$$
	Therefore the difference between the sums of the two gadgets is
	$$
	2\left(\sum_{j=0}^p c_jt_j\right)a_i + 2\lambda\left(\sum_{j=0}^p c_j\right)e_{\chi(i)}.
	$$
	Both coefficients are zero by Lemma~\ref{lem:finite-difference}, applied
	to polynomials of degrees $1$ and $0$.  For the common coefficients, we take
	\begin{equation}
		\label{eq:def-gamma}
		\gamma_p=2\sum_{j\in\cP}c_jt_j \quad \text{and} \quad \beta_p=2\lambda\sum_{j\in\cP}c_j=2^p\lambda;
	\end{equation}
	both are positive constants that depend on $p$.
\end{proof}

The following lemma shows that although the local baseline $B_i$ may be large, it cancels when we take the difference between the two bit gadgets.

\begin{lemma}
	\label{lem:pure-response}
	Suppose a final stream sum $s$ has no component in any balancing
	direction $g_c$.  Define
	$$
	B_i=\lambda\ip{e_{\chi(i)}}{s}, \quad \text{and} \quad Z_i=\ip{a_i}{s}.
	$$
	Then for any $i$, $j$ and $\sigma$, we have $\ip{u_{i,j,\sigma}}{s} = B_i+t_jZ_i$ and
	\begin{equation}\label{eq:gadget-difference}
		\sum_{u\in\Gone(a_i)}\ip{u}{s}^p - \sum_{u\in\Gzero(a_i)}\ip{u}{s}^p = \kappa_p Z_i^p,
	\end{equation}
	where $\kappa_p=2p!\,h^p.$
\end{lemma}

\begin{proof}
	Because $s$ has no balancing component,
	$
	\ip{u_{i,j,\sigma}}{s}=B_i+t_jZ_i.
	$
	Since each level appears with both signs, the difference in the left-hand side of
	\eqref{eq:gadget-difference} equals
	$
	2\sum_{j=0}^p c_j(B_i+t_jZ_i)^p,
	$
	which is $2p!\,h^pZ_i^p$ by Lemma~\ref{lem:finite-difference}.
\end{proof}

\subsection{The Hard Input}
\label{sec:Mp-input-reduction}

We now perform an input reduction from the $\textsc{Index}$ problem to $M_p$.
Let $z=(z_1,\ldots,z_N)\in\{0,1\}^N$ be Alice's input to
$\textsc{Index}_N$.  She inserts the multiset of vectors
$
\biguplus_{i \in [N]} \mathcal{G}_{z_i}(a_i).
$
By Lemma~\ref{lem:equal-sums}, its vector sum is independent of $z$ :
\begin{equation}
	\label{eq:sA}
	s_A = \gamma_p\sum_{i \in [N]} a_i + \beta_p\sum_{i \in [N]} e_{\chi(i)}.
\end{equation}
In particular, $s_A$ has no balancing component.

Bob receives an index $q\in[N]$.  In a fresh $m$-dimensional block,
let $w_1,\ldots,w_m$ be an orthonormal basis and define vectors
\begin{equation*}
	\label{eq:vell}
	v_\ell = \frac1{\sqrt m}a_q + \sqrt{1-\frac1m}\,w_\ell, \qquad \ell\in[m].
\end{equation*}
Each $v_\ell$ is a unit vector, and
\begin{equation*}
	\label{eq:suffix-ip}
	\ip{v_\ell}{v_r}=
	\begin{cases}
		1,& \text{if } \ell=r,\\
		1/m,& \text{if } \ell\neq r.
	\end{cases}
\end{equation*}
Furthermore,
\begin{equation*}
	\label{eq:suffix-prefix-ip}
	\ip{v_\ell}{u_{i,j,\sigma}} = \frac{t_j}{\sqrt m}\ip{a_q}{a_i} \geq 0.
\end{equation*}
Thus, Bob's vectors preserve the nonnegative-cosine promise.

Bob appends $L$ copies of every $v_\ell$, where
$L=\left\lfloor m^{(p-2)/2} \right\rfloor$, and defines the query amplitude
$Q:=L\sqrt m$. The code-space component of Bob's suffix sum is
$Qa_q$. Consequently, the final code response in the $i$-th gadget is
\begin{equation}
	\label{eq:Zi-formula}
	Z_i = \gamma_p\sum_{r\in[N]}\ip{a_i}{a_r} + Q\ip{a_i}{a_q}.
\end{equation}
The first term is the response created by Alice's prefix, while the
second is the additional response created by Bob's query suffix. We
refer to $Z_q$ as the \emph{queried code response} and to $Z_i\ (i \neq q)$ as the \emph{nonqueried code responses}. Since $\ip{a_q}{a_q}=1$ and all codeword inner products are nonnegative,
\begin{equation}
	\label{eq:Zq-lower}
	Z_q\geq Q.
\end{equation}
For $i\neq q$, Bob's suffix contributes only $Q\ip{a_i}{a_q}$; the resulting nonqueried responses are the ``noise'' controlled by the sparse-code properties. Finally, the stream sum has no balancing component, so Lemma~\ref{lem:pure-response} applies.

The whole stream can be written as 
\begin{equation}
	\label{eq:P}
	P(z,q) = \left(
	\biguplus_{i\in[N]} \mathcal G_{z_i}(a_i),
	\underbrace{v_1,\ldots,v_1}_{L\text{ copies}},
	\underbrace{v_2,\ldots,v_2}_{L\text{ copies}},
	\ldots,
	\underbrace{v_m,\ldots,v_m}_{L\text{ copies}}
	\right).
\end{equation}

\subsection{Moment Upper Bounds}

Set $R_{\max}=\left(N/m\right)^{(p-1)/(p-2)}.$ By Lemma~\ref{lem:code}, there is a sufficiently large constant $C_R$ such that, for every $i\in[N]$,
\begin{equation}
	\label{eq:row-one-bound}
	R_i:=1+\sum_{r\neq i}\ip{a_i}{a_r} \leq C_R R_{\max},
\end{equation}
and for every $q$,
\begin{equation}
	\label{eq:row-p-bound}
	\sum_{i\neq q}\ip{a_i}{a_q}^p\leq\iota.
\end{equation}

The exponent $\rho_p$ was chosen to balance Alice's total code contribution $\sum_{i \in [N]}  R_i^p = O_p(N R_{\max}^p)$ with Bob's query contribution $\Omega_p\left(Z_q^p\right) = \Omega_p\left(Q^p\right)$.   Note that
\begin{eqnarray*}
	\label{eq:exponent-balance}
	\frac{N R_{\max}^p}{Q^p} 
	&=&\Theta\left(\frac{N\left(N/m\right)^{p(p-1)/(p-2)}}{m^{p(p-1)/2}}\right) = \Theta_p \left(\frac{c_N m^{\rho_p}\left(c_N m^{\rho_p}/m\right)^{p(p-1)/(p-2)}}{m^{p(p-1)/2}}\right)\\
	&=&\Theta_p\left({c_N}^{(p^2-2)/(p-2)}\right).
\end{eqnarray*}
Consequently, after the code parameter $\iota$ is fixed, we may make
$N R_{\max}^p/Q^p$ arbitrarily small by choosing $c_N$ sufficiently small.
\smallskip

The following lemma shows that the total contribution of the
nonqueried code responses is small compared with the queried signal
scale $Q^p$.
\begin{lemma}
	\label{lem:interference}
	For every constant $\nu>0$, the code parameter $\iota$ and then the
	constant $c_N$ in the definition of $N$ can be chosen so that,
	for every query index $q$, $\sum_{i\neq q} Z_i^p\leq \nu Q^p.$
\end{lemma}

\begin{proof}
	By \eqref{eq:Zi-formula} and the fact $(x+y)^p\leq2^{p-1}(x^p+y^p)$,
	$$
	\sum_{i\neq q}Z_i^p \leq 2^{p-1}\gamma_p^p\sum_{i\neq q}R_i^p + 2^{p-1}Q^p\sum_{i\neq q}\ip{a_i}{a_q}^p.
	$$
	Combining the previous inequality with \eqref{eq:row-one-bound}, \eqref{eq:row-p-bound}, and recalling that $\gamma_p$ is a constant, we obtain
	$$
	\sum_{i\neq q}Z_i^p \leq C_B N R_{\max}^p+2^{p-1}\iota Q^p,
	$$
	where $C_B = C_B(p, \iota)$ is a large enough constant.
	
	We first choose $\iota$ so that the second term is at most $\nu Q^p/2$, and then choose $c_N$ (in the definition of $N$) sufficiently small that $C_B N R_{\max}^p\leq\nu Q^p/2$. The lemma follows.
\end{proof}

We next bound the contribution of the all-zero instance, in which every Alice gadget is a zero-gadget. For a query index $q$, let the {\em global baseline} $V_0(q)$ be the value of $\Mp$ on the stream consisting of $\Gzero(a_i)$ for every $i\in[N]$, followed by Bob's
suffix for $q$.  Because the two gadgets have equal vector sums,
$V_0(q)$ is independent of Alice's actual bit string.

\begin{lemma}
	\label{lem:baseline}
	There is a constant $C_V=C_V(p,\nu)$ such that, for all sufficiently large
	$m$ and every query index $q$, it holds that $0 \leq V_0(q)\leq C_V Q^p.$
\end{lemma}

\begin{proof}
	We separately bound Alice's prefix vector densities and Bob's suffix vector densities.
	
	Each zero-gadget contains only $O_p(1)$ vectors, and every one has density at most $B_i+t_pZ_i$.  Therefore, Alice's total
	contribution in the all-zero instance is at most
	$
	O_p\left(\sum_{i \in [N]}(B_i^p+Z_i^p)\right).
	$
	
	For a color $c$, let $n_c=|\chi^{-1}(c)| = N/m$.  By \eqref{eq:sA} and the definition of $\beta_p$ in \eqref{eq:def-gamma}, the anchor part of a vector's density is
	$
	B_i=\lambda\beta_p n_{\chi(i)}=O_p(N/m).
	$
	Thus, the anchor term satisfies
	$$
	\sum_i B_i^p=O_p\left(N\left(\frac Nm\right)^p\right) = O_p\left((c_Nm^{\rho_p})^{p+1} / m^p \right) = o(m^{p(p-1)/2}) = o(Q^p),
	$$
	where the second-to-last equality follows from the fact that
	$$
	\frac{p(p-1)}2-\bigl((p+1)\rho_p-p\bigr)=\frac{p(p-2)(p+1)}{2(p^2-2)}>0.
	$$
	For the code term, \eqref{eq:Zi-formula} gives
	$$
	\sum_iZ_i^p \leq 2^{p-1}\gamma_p^p\sum_iR_i^p + 2^{p-1}Q^p\sum_i\ip{a_i}{a_q}^p.
	$$
	The first term is $O_p(NR_{\max}^p)=O_p(Q^p)$, and the second is at most
	$2^{p-1}(1+\iota)Q^p = O_p(Q^p)$, including the self-correlation at $i=q$.  
	Combining the anchor term and the code term, the prefix contribution is 
	$O_p(Q^p)$.
	
	It remains to bound the suffix.  For any copy of $v_\ell$, its density is
	\begin{equation*}\label{eq:suffix-density}
		\ip{v_\ell}{s_A} + L\ip{v_\ell}{\sum_{r \in [m]} v_r} = \frac{\gamma_p}{\sqrt m}R_q + L\left(2-\frac1m\right).
	\end{equation*}
	The first term is $O_p(R_{\max}/\sqrt m)$ by
	\eqref{eq:row-one-bound}.  At our choices of parameters,
	$
	{R_{\max}}/{\sqrt m}=o\left(m^{(p-2)/2}\right) =o(L).
	$
	Indeed, the exponent of $R_{\max}$ is
	$$
	(\rho_p-1)\frac{p-1}{p-2}=\frac{(p-2)(p-1)(p+1)}{2(p^2-2)},
	$$
	and
	$$
	\frac{p-2}{2}-\left(\frac{(p-2)(p-1)(p+1)}{2(p^2-2)}-\frac12\right)=\frac{p(p-1)}{2(p^2-2)}>0.
	$$
	Thus every suffix vector has density $O_p(L)$.  There are $mL$ suffix vectors, so their total contribution is $O_p(mL^{p+1})=O_p(Q^p)$ when $p > 2$.
	
	Combining the prefix and suffix parts proves the lemma.
\end{proof}

\subsection{The Reduction}

Let $P(z,q)$ in \eqref{eq:P} denote the stream of vectors generated from Alice's bit
string $z$ and Bob's index $q$.  Since Alice's prefix sum is
independent of $z$, all quantities $B_i$, $Z_i$ (recall their definitions in Lemma~\ref{lem:pure-response}), and all suffix
densities are independent of $z$.

The following lemma shows that the $p$-th density moment is the global baseline $V_0(q)$ plus an independent contribution from each bit with $z_i=1$.
\begin{lemma}
	\label{lem:decomposition}
	For every $z\in\{0,1\}^N$ and every $q\in[N]$,
	$
	\Mp(P(z,q)) = V_0(q)+\kappa_p\sum_{i:z_i=1}Z_i^p.
	$
\end{lemma}

\begin{proof}
	We start from the all-zero instance defining $V_0(q)$.  By Lemma~\ref{lem:equal-sums}, replacing the $i$-th zero-gadget by the one-gadget does not change the final stream
	sum, and therefore does not change any other vector's density.  Lemma~\ref{lem:pure-response} shows that the $i$-th replacement increases the moment by exactly $\kappa_pZ_i^p$. Summing over the indices with $z_i=1$ proves the lemma.
\end{proof}

Choose $\nu=1/10$ in Lemma~\ref{lem:interference}.  If $z_q=0$, then Lemma~\ref{lem:interference} gives
\begin{equation}
	\label{eq:case-zero}
	\Mp(P(z,q)) \leq V_0(q)+\kappa_p\nu Q^p.
\end{equation}
If $z_q=1$, then \eqref{eq:Zq-lower} gives
\begin{equation}
	\label{eq:case-one}
	\Mp(P(z,q)) \geq V_0(q)+\kappa_p Q^p.
\end{equation}
By Lemma~\ref{lem:baseline}, we know ${V_0(q)}/{Q^p}\in[0,C_V].$
We choose
$
\varphi_p := \frac{C_V + \kappa_p}{C_V + \kappa_p\nu}  > 1.
$
Thus, 
$$
\frac{V_0(q)+\kappa_p Q^p}{V_0(q)+\kappa_p\nu Q^p} \ge \varphi_p.
$$
Consequently, for any $\eps_p \in \left(0, \frac{\varphi_p-1}{\varphi_p+1}\right)$, a $(1+\eps_p)$-approximation to $M_p(P(z, q))$ can be used to distinguish the two cases $z_q = 0$ and $z_q = 1$.

\begin{proof}[Proof of Theorem~\ref{thm:main}]
	Alice runs the streaming algorithm on her prefix vectors determined by input $z$ and sends its memory configuration to Bob.  Bob continues the algorithm on the suffix vector set determined by input $q$ and uses the algorithm's estimate of $M_p(P(z,q))$ to recover $z_q$ with the same constant success probability.  Lemma~\ref{lem:index-lower-bound} implies that we need at least $\Omega(N) = \Omega(m^{\rho_p})$ bits of memory.
	
	The four orthogonal blocks have total dimension $4m$.  The finite-precision implementation to be described in Section~\ref{sec:finite-precision} uses one additional common coordinate, for a total of $4m+1$ dimensions.  For an arbitrary sufficiently large dimension $d$, set
	$m=\lfloor(d-1)/4\rfloor$; the space bound becomes $\Omega_p(d^{\rho_p}).$
	In the input reduction in Section~\ref{sec:Mp-input-reduction},	Alice inserts $2^pN=O_p(N)$ vectors,  and Bob inserts $mL=O(m^{p/2})$ vectors.  Since $p$ is fixed, the total stream length is polynomial in $d$.  We will show in Section~\ref{sec:finite-precision} that the vectors can be
	chosen as exact rational unit vectors with $O_p(\log d)$ bits per
	coordinate while preserving the constant gap $\varphi_p$.
\end{proof}

\subsection{Finite Precision}
\label{sec:finite-precision}

The above construction uses real unit vectors and guarantees nonnegative
pairwise similarities, some of which may be zero. A direct rational
approximation could make these similarities slightly negative. We
therefore first introduce a small positive margin and then obtain exact
rational unit vectors with $O_p(\log m)$ bits per coordinate while
preserving the nonnegative-cosine promise.

Let $n$ be the length of the hard input stream.  From the construction,
$n=2^pN+mL=O_p(m^{p/2}).$ Choose $\xi=m^{-2p}$ and map every ideal unit vector $u\in\R^{4m}$ to
$
\tilde u=\left(\sqrt{1-\xi^2}\,u,\xi\right)\in\R^{4m+1}.
$
Then $\tilde u$ is a unit vector and
$$
\ip{\tilde u}{\tilde v}=(1-\xi^2)\ip{u}{v}+\xi^2\geq\xi^2.
$$
Thus, every pairwise similarity is strictly positive.  If $D_i$ is an
ideal density and $\widetilde D_i$ is the corresponding perturbed
density, then
$
\widetilde D_i=(1-\xi^2)D_i+n\xi^2.
$
Since $D_i\geq1$ and $n\xi^2=o(1)$, we have $\widetilde D_i=(1+o(1))D_i$ and 
$$
\sum_i\widetilde D_i^p=(1+o(1))\sum_iD_i^p.
$$

We next replace each $\tilde u$ by a nearby rational unit vector. We use the following fact.

\begin{lemma}
	\label{lem:rational-sphere}
	Let $\phi\geq 2$, let $x=(x',x_\phi)\in\R^{\phi-1}\times\R$ satisfy $\norm{x}=1$ and $x_\phi>0$, and let $\eta\in(0,1)$. There exists an exact rational unit vector $\hat{x}\in\mathbb Q^\phi$ such that $\norm{\hat{x}}=1$, $\hat{x}_\phi>0$, and $\norm{\hat{x}-x}\leq\eta.$ Moreover, the numerator and denominator of every coordinate of $\hat{x}$ can be represented using
	$
	O\!\left(\log \phi + \log\frac{1}{\eta} + \log\frac{1}{x_\phi}\right)
	$
	bits.
\end{lemma}

\begin{proof}
	Consider the rational parametrization of the unit sphere
	$$
	\Psi(y) := \left(\frac{2y}{1+\norm{y}^2}, \frac{1-\norm{y}^2}{1+\norm{y}^2}\right),
	\qquad y \in \R^{\phi-1}.
	$$
	A direct calculation shows that $\norm{\Psi(y)}=1$.
	Moreover, for all $y,z\in\R^{\phi-1}$,
	\begin{equation}
		\label{eq:stereographic-distance}
		\norm{\Psi(y)-\Psi(z)}^2 = \frac{4\norm{y-z}^2} {\bigl(1+\norm{y}^2\bigr)\bigl(1+\norm{z}^2\bigr)} \leq 4\norm{y-z}^2.
	\end{equation}
	
	Since $x_\phi>0$, the vector $y:=\frac{x'}{1+x_\phi}$ belongs to the open unit ball and satisfies $\Psi(y)=x$. Indeed,
	$$
	\norm{y}^2 = \frac{1-x_\phi^2}{(1+x_\phi)^2} = \frac{1-x_\phi}{1+x_\phi} < 1.
	$$
	In addition,
	\begin{equation}
		\label{eq:stereographic-margin}
		1-\norm{y} = \frac{1-\norm{y}^2}{1+\norm{y}} = \frac{2x_\phi}{(1+x_\phi)(1+\norm{y})} \geq
		\frac{x_\phi}{2}.
	\end{equation}
	
	Set $\tau := \min\left\{\frac{\eta}{2}, \frac{x_\phi}{4}\right\}$, and $q := \left\lceil \frac{\sqrt{\phi-1}}{2\tau}\right\rceil.$ For each $j\in[\phi-1]$, let $r_j$ be an integer nearest to
	$qy_j$, and define
	$$
	\hat{y} := \frac{1}{q}(r_1,\ldots,r_{\phi-1}) \in\mathbb Q^{\phi-1}.
	$$
	Then
	$$
	\norm{\hat{y}-y} \leq \frac{\sqrt{\phi-1}}{2q} \leq \tau.
	$$
	By \eqref{eq:stereographic-margin} and the choice of $\tau$,
	$$
	\norm{\hat{y}} \leq \norm{y}+\tau \leq \norm{y}+\frac{1-\norm{y}}{2} < 1.
	$$
	Hence, the last coordinate of $\Psi(\hat{y})$ is positive. Define $\hat{x}:=\Psi(\hat{y}).$ Because $\hat{y}$ is rational, $\hat{x}$ is a rational unit vector. Furthermore, by \eqref{eq:stereographic-distance},
	$$
	\norm{\hat{x}-x} = \norm{\Psi(\hat{y})-\Psi(y)} \leq 2\norm{\hat{y}-y} \leq 2\tau \leq \eta.
	$$
	
	It remains to verify the bit complexity. Write $r=(r_1,\ldots,r_{\phi-1})$ and
	$
	S:=q^2+\sum_{j=1}^{\phi-1}r_j^2.
	$
	Substituting $\hat{y}=r/q$ into the definition of $\Psi$ gives
	$$
	\hat{x} = \left(\frac{2qr_1}{S}, \ldots, \frac{2qr_{\phi-1}}{S}, \frac{q^2-\sum_{j=1}^{\phi-1}r_j^2}{S}\right).
	$$
	Since $\norm{y}<1$, each $|r_j|$ is at most $q+1$. Therefore, $S=O(\phi q^2),$
	and every numerator and denominator in the preceding display has $O(\log \phi+\log q)$ bits. Finally,
	$
	q = O\left(\frac{\sqrt \phi}{\min\{\eta,x_\phi\}}\right),
	$
	so
	$$
	\log \phi+\log q = O\left(\log \phi + \log\frac{1}{\eta} + \log\frac{1}{x_\phi}\right).
	$$
	This proves the lemma.
\end{proof}

Apply Lemma~\ref{lem:rational-sphere} to every perturbed vector
$\tilde u\in\R^{4m+1}$ with $\eta=m^{-6p}$, and let
$\hat u\in\mathbb Q^{4m+1}$ be the resulting rational unit vector.
Since the last coordinate of $\tilde u$ is $\xi=m^{-2p}$, every
coordinate of $\hat u$ has $O_p(\log m)$ bits. Moreover, for any two vectors,
$$
\left|\ip{\hat u}{\hat v} - \ip{\tilde u}{\tilde v} \right| \leq \norm{\hat u-\tilde u}_2 + \norm{\hat v-\tilde v}_2 \leq 2\eta.
$$
Every perturbed inner product is at least $\xi^2$, and $2\eta=o(\xi^2)$, so all rationalized pairwise similarities remain positive for sufficiently large $m$. Consequently, if $\widetilde D_i$ and $\widehat D_i$ denote the densities before and after rationalization, then
$$
\left|\widehat D_i-\widetilde D_i\right| \leq 2n\eta=o(1).
$$
Since $\widetilde D_i\geq1$, uniformly over Alice's input and Bob's query, $\widehat D_i=(1+o(1))\widetilde D_i$ and hence,
$$
\sum_i\widehat D_i^p = (1+o(1))\sum_i\widetilde D_i^p = (1+o(1))\sum_iD_i^p.
$$

Suppose that the ratio between the one-case and zero-case objective
values in the ideal construction is at least $\varphi_p>1$. The
preceding uniform estimate implies that, after perturbation and
rationalization, this ratio is at least
$
\varphi_p \cdot \frac{1-o(1)}{1+o(1)} = \varphi_p-o(1).
$
Thus, for all sufficiently large $m$, it is bounded below by some constant $\tilde\varphi_p>1$. Choose a constant $\eps_p>0$ such that
$
\frac{1+\eps_p}{1-\eps_p} < \tilde\varphi_p.
$
Then the approximation intervals in the two cases are disjoint, so a
$(1+\eps_p)$-approximation still recovers the queried bit. Therefore,
the communication reduction applies to the rationalized hard
instances.

\section{Conclusion}
\label{sec:conclusion}

In this paper, we studied the streaming complexity of cosine density moments, which extend classical frequency moments from equality to vector data under cosine similarity. Our results show that the low-rank Gram structure of cosine similarity permits dimension-dependent streaming algorithms despite the impossibility results for general similarity functions.

The main remaining question is to determine the optimal complexity of higher
density moments under nonnegative cosine. For each fixed integer $p>2$, our
upper and lower bounds have dimension exponents $p/2$ and
$
\rho_p=\frac{p^2(p-1)}{2(p^2-2)},
$
respectively, where $(p-1)/2<\rho_p<p/2$. Closing this gap and determining the optimal dependence on $\eps$ remain open.  More generally, it would be interesting to develop a theory characterizing which structural properties of a similarity function allow for one-pass, sublinear-space streaming algorithms.

\section*{Use of AI assistance}
The research questions, algorithms, proof strategies, and substantive mathematical arguments in this paper were developed by the author. ChatGPT (GPT-5.5 and GPT-5.6) was used for language editing and proofreading, and to assist with parameter optimization in the gadget constructions and routine algebraic manipulation. The resulting text and calculations were verified by the author, who takes full responsibility for the paper's contents.

\bibliographystyle{alpha}
\bibliography{paper}

\appendix

\section{Math Tools}
\label{app:math-tools}

\begin{lemma}[Bernstein's inequality]
	\label{lem:Bernstein}
	Let $X_1,\ldots,X_m$ be independent mean-zero random variables satisfying $|X_i|\le b$ almost surely, and let
	$
		\sigma^2:=\sum_{i \in [m]} \E[X_i^2].
	$
	Then, for every $t>0$,
	$$
		\Pr\left[ \Biggl|\sum_{i \in [m]} X_i\Biggr|\ge t \right] \le 2\exp\left(-\frac{t^2}{2\sigma^2+\frac{2}{3}bt} \right).
	$$
	In particular, for every $x>0$,
	$$
		\Pr\left[\Biggl|\sum_{i \in [m]} X_i\Biggr| \ge \sqrt{2\sigma^2x}+\frac{2}{3}bx \right] \le 2e^{-x}.
	$$
	For a one-sided application, the leading factor $2$ on the right-hand side of both inequalities can be omitted.
\end{lemma}

\section{Auxiliary Lemmas and Proofs}
\label{app:acute-auxiliary-proof}

\subsection{A Sparse-Intersection Moment Bound}
\label{sec:lem:sparse-intersection}

\begin{lemma}
	\label{lem:hypergeometric-moment}
	Let $A$ and $B$ be independent uniformly random $s$-subsets of
	$[m]$, and let $X:=|A\cap B|$.
	Fix an integer $p\geq 1$. If $m\geq 2p$ and ${s^2}/{m}\leq 1$,
	then
	$
	\E[X^p] \leq C_p \cdot {s^2}/{m},
	$
	where the constant $C_p>0$ depends only on $p$. 
\end{lemma}

\begin{proof}
	For an integer-valued random variable $Y$, write
	$
	(Y)_r:=Y(Y-1)\cdots(Y-r+1)
	$
	for the falling factorial, with $(Y)_0:=1$.
	The ordinary power $X^p$ has the falling-factorial expansion
	$$
	X^p=\sum_{r \in [p]}\stirling{p}{r}(X)_r.
	$$
	All coefficients $\stirling{p}{r}$ are nonnegative and depend only on
	$p$. Taking expectations on both sides gives
	\begin{equation}
		\label{eq:Xp}
		\E[X^p]=\sum_{r \in [p]}\stirling{p}{r}\E[(X)_r].
	\end{equation}
	
	We first compute the factorial moments of $X$. The random variable
	$(X)_r$ counts the number of ordered $r$-tuples of distinct elements
	contained in $A\cap B$. There are $(m)_r$ ordered $r$-tuples of
	distinct elements in $[m]$. For any fixed such tuple $T$,
	$$
	\Pr[T \subseteq A]=\frac{(s)_r}{(m)_r}.
	$$
	Since $A$ and $B$ are independent,
	$$
	\Pr[T \subseteq A\cap B]=\left(\frac{(s)_r}{(m)_r}\right)^2.
	$$
	By linearity of expectation,
	\begin{equation}
		\label{eq:Xr}
		\E[(X)_r]=(m)_r\left(\frac{(s)_r}{(m)_r}\right)^2=\frac{(s)_r^2}{(m)_r}.
	\end{equation}
	
	For every $1\leq r\leq p$, we have $(s)_r\leq s^r.$
	Moreover, since $m\geq 2p$,
	\begin{equation}
		\label{eq:mr}
		(m)_r=\prod_{k=0}^{r-1}(m-k)\geq\left(\frac{m}{2}\right)^r.
	\end{equation}
	Combining \eqref{eq:Xr} and \eqref{eq:mr}, we have
	$
	\E[(X)_r]\leq2^r\left({s^2}/{m}\right)^r.
	$
	Since $s^2/m \le 1$, 
	$\E[(X)_r]\leq 2^r \cdot {s^2}/{m}.$
	Substituting this into \eqref{eq:Xp} gives
	$$
	\E[X^p]\leq \left(\sum_{r \in [p]}\stirling{p}{r}2^r\right)
	\frac{s^2}{m}.
	$$
	Setting $C_p = \sum_{r \in [p]}\stirling{p}{r}2^r$ proves the lemma.
\end{proof}

\subsection{Proof of Lemma~\ref{lem:finite-difference}}
\label{sec:proof:lem:finite-difference}

\begin{proof}
	For a function $f\colon\mathbb{R}\to\mathbb{R}$ and a step size
	$h > 0$, define the forward-difference operator
	$$
	(\Delta_h f)(t):=f(t+h)-f(t).
	$$
	Its $p$-fold iterate satisfies
	\begin{equation}
		\Delta_h^p f(t) = \sum_{j=0}^{p}(-1)^{p-j}\binom{p}{j}f(t+jh).
		\label{eq:forward-difference-expansion}
	\end{equation}
	Indeed, if $E_h f(t):=f(t+h)$, then $\Delta_h=E_h-I$, and hence
	$$
	\Delta_h^p = (E_h-I)^p = \sum_{j=0}^{p}(-1)^{p-j}\binom{p}{j}E_h^j,
	$$
	which gives \eqref{eq:forward-difference-expansion}.
	
	For $0\leq r\leq p$, define
	$
	f_r(t):=(B+tZ)^r.
	$
	Using $c_j=(-1)^{p-j}\binom{p}{j}$ and $t_j=t_0+jh$, we obtain
	\begin{equation}
		\label{eq:fr-expand}
		\sum_{j=0}^{p}c_j(B+t_jZ)^r = \sum_{j=0}^{p}(-1)^{p-j}\binom{p}{j} f_r(t_0+jh) = \Delta_h^p f_r(t_0).
	\end{equation}
	
	We use the following standard property of the forward-difference
	operator on polynomials. Let $g$ be a polynomial of degree $k\geq 1$
	with leading coefficient $a_k$, and write
	$g(t)=a_kt^k+q(t)$,
	where
	$
	\deg(q)\leq k-1.
	$
	Then
	\begin{equation*}
		(\Delta_h g)(t) = g(t+h)-g(t) = a_k\bigl((t+h)^k-t^k\bigr)+(\Delta_h q)(t) = kha_kt^{k-1}+\widetilde q(t),
	\end{equation*}
	where $\deg(\widetilde q)\leq k-2$. Indeed,
	$$
	(t+h)^k-t^k = kht^{k-1} + \sum_{\ell=0}^{k-2}\binom{k}{\ell}t^\ell h^{k-\ell},
	$$
	and $\deg(\Delta_h q)\leq k-2$. Thus, since $h>0$, applying
	$\Delta_h$ to a degree-$k$ polynomial lowers its degree by exactly one
	and multiplies its leading coefficient by $kh$. Iterating this
	process, $\Delta_h^p$ cancels every polynomial of degree less
	than $p$, while
	$
	\Delta_h^p t^p=p!\,h^p.
	$
	
	If $r<p$, then $f_r$ has degree at most $r<p$, and therefore
	\begin{equation}
		\label{eq:fr}
		\Delta_h^p f_r(t_0)=0.
	\end{equation}
	If $r=p$, then the coefficient of $t^p$ in
	$f_p(t)=(B+tZ)^p$ is $Z^p$. All remaining terms have degree less than
	$p$, so by linearity,
	\begin{equation}
		\label{eq:fp}
		\Delta_h^p f_p(t_0) = Z^p\Delta_h^p t^p = p!\,h^pZ^p.
	\end{equation}
	Combining \eqref{eq:fr} and \eqref{eq:fp} with
	\eqref{eq:fr-expand} proves the lemma.
\end{proof}

\end{document}